\documentclass[draftcls,11pt]{IEEEtran}
\onecolumn 
\usepackage[cmex10]{amsmath}
\usepackage{xcolor}
\usepackage[multiple]{footmisc}

\usepackage{mdwtab}
\usepackage{eqparbox}
\usepackage{tabularx}
\usepackage{graphicx,amssymb,rangecite,upgreek,graphicx,dsfont,mathrsfs}

\usepackage{algorithm}
\usepackage{algorithmic} 
\usepackage[usenames,dvipsnames]{pstricks}
 \usepackage{epsfig}

 \usepackage{pst-grad} 
 \usepackage{pst-plot} 

\usepackage[T1]{fontenc}
\usepackage{mwe}    
\usepackage{subfig}
\allowdisplaybreaks
\newcommand {\exe} {\stackrel{\cdot} {=}}

\newcommand {\bh} {\mbox{\boldmath $h$}}

\newcommand {\bs} {\mbox{\boldmath $s$}}
\newcommand {\bst} {\mbox{\footnotesize\boldmath $s$}}
\newcommand {\bvt} {\mbox{\footnotesize\boldmath $v$}}
\newcommand {\but} {\mbox{\footnotesize\boldmath $u$}}

\newcommand {\bv} {\mbox{\boldmath $v$}}

\newcommand {\bx} {\mbox{\boldmath $x$}}
\newcommand {\by} {\mbox{\boldmath $y$}}

\newcommand {\bA} {\mbox{\boldmath $A$}}

\newcommand {\bE} {\mathbb{E}}
\newcommand {\pr} {\mathbb{P}}

\newcommand {\bH} {\mbox{\boldmath $H$}}
\newcommand {\bI} {\mbox{\boldmath $I$}}

\newcommand {\bpi} {\mbox{\boldmath $\Pi$}}
\newcommand {\bS} {\mbox{\boldmath $S$}}

\newcommand {\bU} {\mbox{\boldmath $U$}}
\newcommand {\calHs} {\mbox{\boldmath $\calH^{\bst}$}}

\newcommand {\bV} {\mbox{\boldmath $V$}}
\newcommand {\bW} {\mbox{\boldmath $W$}}
\newcommand {\bX} {\mbox{\boldmath $X$}}
\newcommand {\bY} {\mbox{\boldmath $Y$}}

\newcommand{\bHtt}{\bH_{\bst}}

\newcommand{\bItt}{\bI_{\bst}}

\newcommand{\calA}{{\cal A}}
\newcommand{\calB}{{\cal B}}

\newcommand{\calD}{{\cal D}}

\newcommand{\calF}{{\cal F}}
\newcommand{\calG}{{\cal G}}
\newcommand{\calH}{{\cal H}}
\newcommand{\calI}{{\cal I}}

\newcommand{\calN}{{\cal N}}
\newcommand{\calO}{{\cal O}}

\newcommand{\calR}{{\cal R}}
\newcommand{\calS}{{\cal S}}

\newcommand{\calV}{{\cal V}}

\newcommand{\define}{\stackrel{\triangle}{=}}

\newcommand{\be}{\begin{equation}}
\newcommand{\ee}{\end{equation}}
\newcommand{\beqna}{\begin{eqnarray}}
\newcommand{\eeqna}{\end{eqnarray}}

\usepackage{theorem}
\DeclareFontFamily{U}{mathx}{\hyphenchar\font45}
\DeclareFontShape{U}{mathx}{m}{n}{
      <5> <6> <7> <8> <9> <10>
      <10.95> <12> <14.4> <17.28> <20.74> <24.88>
      mathx10
      }{}
\DeclareSymbolFont{mathx}{U}{mathx}{m}{n}
\DeclareMathSymbol{\bigtimes}{1}{mathx}{"91}
\theorembodyfont{\rmfamily}

\newcommand{\abs}[1]{\left|#1\right|}

\theoremheaderfont{\itshape}

\newtheorem{theorem}{Theorem}
\newtheorem{proof}{Proof}

\newtheorem{claim}{Claim}

\newtheorem{lemma}{Lemma}

\newcommand{\p}[1]{\left(#1\right)}
\newcommand{\pp}[1]{\left[#1\right]}
\newcommand{\ppp}[1]{\left\{#1\right\}}
\newcommand{\norm}[1]{\left\|#1\right\|}
\def\bpi{{\mbox{\boldmath $\Pi$}}}

\begin{document}

\title{On Compressive Sensing in Coding Problems: A Rigorous Approach}
\author{Wasim~Huleihel,
				~Neri~Merhav,
				~Shlomo~Shamai~(Shitz)
				\\
        Department of Electrical Engineering \\
Technion - Israel Institute of Technology \\
Haifa 32000, ISRAEL\\
E-mail: \{wh@tx, merhav@ee, sshlomo@ee\}.technion.ac.il
\thanks{$^\ast$The work of Huleihel and Merhav was partially supported by The Israeli Science Foundation (ISF), Grant no. 412/12. The work of Shamai was supported by The Israeli Science Foundation (ISF), the European Commission in the framework of the FP7 Network of Excellence in Wireless COMmunications NEWCOM\# and  by S. and N. Grand Research Fund.}
}
\maketitle

\IEEEpeerreviewmaketitle

\begin{abstract}
\boldmath We take an information theoretic perspective on a classical sparse-sampling noisy linear model and present an analytical expression for the mutual information, which plays central role in a variety of communications/processing problems. Such an expression was addressed previously either by bounds, by simulations and by the (non-rigorous) replica method. The expression of the mutual information is based on techniques used in \cite{Wasim2}, addressing the minimum mean square error (MMSE) analysis. Using these expressions, we study specifically a variety of sparse linear communications models which include coding in different settings, accounting also for multiple access channels and different wiretap problems. For those, we provide single-letter expressions and derive achievable rates, capturing the communications/processing features of these timely models.
\end{abstract}


\section{Introduction}

\IEEEPARstart{C}{ompressed} sensing \cite{Tao,Donho1} is a collection of signal processing techniques that compress sparse analog vectors by means of linear transformations. Using some prior knowledge on the signal \emph{sparsity}, and by designing efficient encoders and decoders, the goal is to achieve effective compression in the sense of taking a much smaller number of measurements than the dimension of the original signal. Recently, a vast amount of research was conducted concerning sparse random Gaussian signals which are very relevant to wireless communications, see, for example, \cite{Wasim2,WuVerdu,Gastpar1,Tulino} and many references therein. 

A general setup of compressed sensing is shown in Fig. \ref{fig:Moisycompressed}. The mechanism is as follows: A real vector $\bX\in\mathbb{R}^n$ is mapped into $\bV\in\mathbb{R}^k$ by an encoder (or compressor) $f:\mathbb{R}^n\to\mathbb{R}^k$. The decoder (decompressor) $g:\mathbb{R}^k\to\mathbb{R}^n$ receives $\bY$, which is a noisy version of $\bV$, and outputs $\hat{\bX}$ as the estimation of $\bX$. The sampling rate, or the compression ratio, is defined as
\begin{align}
q \define \frac{k}{n}.
\end{align}
In this paper, the encoder is constrained to be a \emph{linear} mapping, denoted by a matrix $\bH\in\mathbb{R}^{k\times n}$, usually called the \emph{sensing matrix} or \emph{measurement matrix}, where $\bH$ is assumed to be a random matrix with i.i.d. entries of zero mean and variance $1/n$. On the decoder side, most of the compressed sensing literature focuses on low-complexity decoding algorithms which are robust to noise, for example, decoders based on convex optimization, greedy algorithms, etc. (see, for example \cite{Gastpar1,cc6,cc7,cc8}). Although the decoding is, of course, an important issue, it is not in the focus of this work. The input vector $\bX$ is assumed to be random, distributed according some probability density that models the sparsity. Finally, the noise is assumed to additive white and Gaussian. 
\begin{figure}
\centering
\begin{pspicture}(0,-1.2592187)(11.682813,1.2992188)
\psframe[linewidth=0.04,dimen=outer](4.1209373,0.22078125)(1.5009375,-1.2592187)
\psframe[linewidth=0.04,dimen=outer](10.020938,0.22078125)(7.4209375,-1.2592187)
\psellipse[linewidth=0.04,dimen=outer](5.8509374,-0.43921876)(0.53,0.52)
\psline[linewidth=0.04cm,arrowsize=0.05291667cm 2.0,arrowlength=1.4,arrowinset=0.4]{->}(0.3009375,-0.47921875)(1.5609375,-0.49921876)
\psline[linewidth=0.04cm,arrowsize=0.05291667cm 2.0,arrowlength=1.4,arrowinset=0.4]{->}(6.4009376,-0.43921876)(7.3809376,-0.45921874)
\psline[linewidth=0.04cm,arrowsize=0.05291667cm 2.0,arrowlength=1.4,arrowinset=0.4]{->}(4.1209373,-0.45921874)(5.3409376,-0.47921875)
\psline[linewidth=0.04cm,arrowsize=0.05291667cm 2.0,arrowlength=1.4,arrowinset=0.4]{->}(10.060938,-0.43921876)(11.320937,-0.45921874)
\psline[linewidth=0.04cm,arrowsize=0.05291667cm 2.0,arrowlength=1.4,arrowinset=0.4]{->}(5.9009376,0.86078125)(5.9209375,0.08078125)
\usefont{T1}{ptm}{m}{n}
\rput(2.7982812,-0.34921876){Encoder}
\usefont{T1}{ptm}{m}{n}
\rput(2.8323438,-0.86921877){$f_n:\;\mathbb{R}^n\to\mathbb{R}^k$}
\usefont{T1}{ptm}{m}{n}
\rput(8.668906,-0.38921875){Decoder}
\usefont{T1}{ptm}{m}{n}
\rput(8.752344,-0.86921877){$g_n:\;\mathbb{R}^k\to\mathbb{R}^n$}
\usefont{T1}{ptm}{m}{n}
\rput(10.712344,-0.16921875){$\hat{\bX}$}
\usefont{T1}{ptm}{m}{n}
\rput(6.8723435,-0.18921874){$\bY$}
\usefont{T1}{ptm}{m}{n}
\rput(5.8523436,1.1107812){$\bW$}
\usefont{T1}{ptm}{m}{n}
\rput(4.742344,-0.18921874){$\bV$}
\usefont{T1}{ptm}{m}{n}
\rput(1.0023438,-0.24921875){$\bX$}
\psline[linewidth=0.04cm](5.8409376,-0.25921875)(5.8409376,-0.6592187)
\psline[linewidth=0.04cm](5.6409373,-0.45921874)(6.0409374,-0.45921874)
\end{pspicture} 
\centering
\caption{Noisy compressed sensing setup.}
\label{fig:Moisycompressed}
\end{figure}
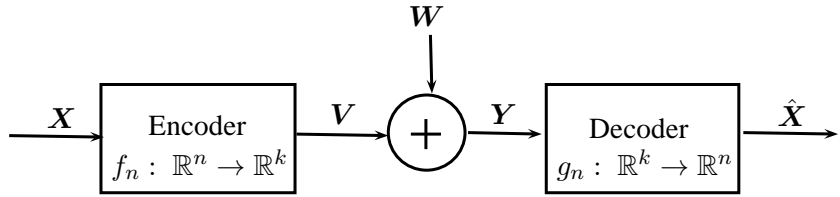

In the literature, there is a great interest in finding asymptotic formulas of some information and estimation measures, e.g., the minimum mean squared error (MMSE), mutual information rates, and other information measures. Finding these formulas is, in general, extremely complicated, and most of the works (e.g., \cite{WuVerdu,Tulino,GuoShamaiBaron}) that deal with this problem resort to using the \emph{replica} method which is borrowed from the field of statistical physics. Although the replica method is powerful, it is non-rigorous. Recently, in \cite{Wasim2} a rigorous derivation of the asymptotic MMSE was carried out, and it was shown that the results obtained support the previously known replica predictions. The key idea in our analysis is the fact that by using some direct relationship between optimum estimation and certain partition functions \cite{Neri1}, the MMSE can be represented in some mathematically convenient form which (due to the previously mentioned input and noise Gaussian statistics assumptions) consists of functionals of the \emph{Stieltjes} and \emph{Shannon} transforms. This observation allows us to use some powerful results from random matrix theory, concerning the asymptotic behavior (a.k.a. deterministic equivalents) of the Stieltjes and Shannon transforms (see e.g., \cite{baisilbook,coulbook} and many references therein). Here, however, we are concerned with some input-output mutual information rates, rather than the asymptotic MMSE. Nonetheless, we show that these information rates are readily obtained from the results of \cite{Wasim2}. It is worthwhile to emphasize that these kind of mutual information rates formulas are useful and important. For example, with relation to this paper, recently, in \cite{TulinoCaire}, the capacity was derived for single-user discrete-time channels subject to both frequency-selective and time-selective fading, where the channel output is observed in additive Gaussian noise. This result is indeed important due to the fact that various mobile wireless systems are subject to both frequency-selective fading and to time-selective fading.

The works cited above focus on uncoded continuous signals, while in this paper, we concentrate on coded communication, similarly to \cite{Peleg}. In other words, we use coded sparse signals, and the objective is to achieve reliable reconstruction of the signal and its support. In \cite{Peleg}, sparse sampling of coded signals at sub-Landau sampling rates was considered. It was shown that with coded and with discrete signals, the Landau condition may be relaxed, and the sampling rates required for signal reconstruction and for support detection can be lower than the effective bandwidth. Equivalently, the number of measurements in the corresponding sparse sensing problem can be smaller than the support size. Tight bounds on information rates and on signal and support detection performance are derived for the Gaussian sparsely sampled channel and for the frequency-sparse channel using the context of state dependent channels. It should be emphasized that part of the coding principles and problems that we will consider in this paper have already appeared in \cite{Peleg}, but relying on bounds. Here, the new results facilitate a rigorous discussion. 

The main goal of this paper is to use the previously mentioned mutual information rates in order to give some new closed-form achievable rates in various channel coding problems, in the wiretap channel model, and in the multiple access channel (MAC). 
Particularity, in the first part of these channel coding problems, we will consider three different cases that differ in the assumptions about the knowledge available at the transmitter and the receivers. For example, in Subsection \ref{Uncontrolled1}, we will consider the case in which the sparsity pattern cannot be controlled by the transmitter, but it is given beforehand. This falls within the well-known framework of state dependent channels \cite{Abbas} (e.g., the Shannon settings \cite{ShannonCaus} and the Gel'fand-Pinsker channel \cite{Gelfand2}). Another interesting result is that when the sparsity pattern is controlled by the transmitter, a memoryless source maximizes the mutual information rate. It is important to comment that this result is attributed to the fact that our mutual information rate formula is valid for sources with memory, which is not the case in previously reported results that were based on the replica method. In the second and third parts of the applications, which deal with the wiretap and the MAC models, respectively, we will consider several cases in the same spirit. For each of these cases, we provide practical motivations and present numerical examples in order to gain some quantitative feeling of what is possible.

The remaining part of this paper is organized as follows. In Section \ref{sec:model}, the model is presented and the problem is formulated. In Section \ref{sec:coding}, the main results concerning channel coding problems are presented and discussed along with a numerical example that demonstrates the theoretical results. In Section \ref{sec:wire}, achievable rates for the wiretap channel model are presented. Then, in Section \ref{sec:MAC}, we present an implication for the MAC, and finally, our conclusions appear in Section \ref{sec:Conclusion}.

\section{Model and Problem Formulation}\label{sec:model}

Consider the following stochastic model: Each component, $X_i$, $1\leq i\leq n$, of $\bX = \p{X_1,\ldots,X_n}$, is given by $X_i = S_iU_i$ where $\ppp{U_i}$ are i.i.d. Gaussian random variables with zero mean and variance $\sigma^2$, and $\ppp{S_i}$ are binary random variables, taking values in $\ppp{0,1}$, independently of $\ppp{U_i}$. Concerning the random vector $\bS = \p{S_1,\ldots,S_n}$ (or, \emph{pattern} sequence), similarly as in \cite{Wasim2}, we postulate that the probability $\pr\p{\bS}$ depends only on the ``\emph{magnetization}"\footnote{The term ``magnetization" is borrowed from the field of statistical mechanics of spin array systems, in which $S_i$ is taking values in $\ppp{-1,1}$. Nevertheless, for the sake of convince, we will use this term also in our problem.}
\begin{align}
m_s \define \frac{1}{n}\sum_{i=1}^nS_i.
\end{align}
In particular, we assume that
\begin{align}
\pr\p{\bS} = C_n\cdot\exp\ppp{nf\p{m_s}}
\label{inputassmeas}
\end{align}
where $f\p{\cdot}$ is a certain function that is independent of $n$, and $C_n$ is a normalization constant. Note that for the customary i.i.d. assumption, $f$ is a linear function. By using the method of types \cite{Cizer}, we obtain\footnote{Throughout this paper, for two positive sequences $\ppp{a_n}$ and $\ppp{b_n}$, the notations $a_n\exe b_n$ and $a_n\approx b_n$ mean equivalence in the exponential order, i.e., $\lim_{n\to\infty}\frac{1}{n}\log\p{a_n/b_n} = 0$, and $\lim_{n\to\infty}\p{a_n/b_n} = 1$, respectively. For two sequences $\ppp{a_n}$ and $\ppp{b_n}$, the notation $a_n\asymp b_n$ means that $\lim_{n\to\infty}\p{a_n-b_n} = 0$.}
\begin{align}
C_n &= \p{\sum_{\bst\in\ppp{0,1}^n}\exp\ppp{nf\p{m_s}}}^{-1}\nonumber\\
&= \p{\sum_{m\in\pp{0,1}}\Omega\p{m}\exp\ppp{nf\p{m}}}^{-1}\nonumber\\
&\exe \exp\ppp{-n\cdot\max_m\ppp{\calH_2\p{m}+f\p{m}}}\label{apriorimag0}\\
& = \exp\ppp{-n\pp{\calH_2\p{m_a}+f\p{m_a}}},
\label{apriorimag}
\end{align} 
where $\Omega\p{m}$ designates the number of binary $n$-vectors with magnetization $m$, $\calH_2\p{\cdot}$ denotes the binary entropy function, and $m_a$ is the maximizer of $\calH_2\p{m}+f\p{m}$ over $\pp{0,1}$. In other words, $m_a$ is the \emph{a-priori} magnetization that \emph{dominates} $\pr\p{\bS}$. Finally, note that in the i.i.d. case, each $X_i$ is distributed according to following mixture distribution (a.k.a. Bernoulli-Gaussian measure)
\begin{align}
P\p{x} = \p{1-p}\cdot\delta\p{x} + p\cdot P_G\p{x}
\label{mes}
\end{align}
where $\delta\p{x}$ is the Dirac function, $P_G\p{x}$ is a Gaussian density function and $0\leq p\leq1$. Then, by the law of large numbers (LLN), $\frac{1}{n}\norm{\bX}_0\stackrel{\mathbb{P}}{\rightarrow}p$, where $\norm{\bX}_0$ designates the number of non-zero elements of a vector $\bX$. Thus, it is clear that the weight $p$ parametrizes the signal sparsity and $P_G$ is the prior distribution of the non-zero entries.

Finally, we consider the following observation model
\begin{align}
\bY = \bA\bH\bX+\bW,
\label{GeneralModel}
\end{align}
where $\bY$ is the observed channel output vector of dimension $n$, $\bA$ is $n\times n$ diagonal matrix with i.i.d. diagonal elements with $\pr\ppp{\bA_{i,i}=1} = q = 1-\pr\ppp{\bA_{i,i}=0}$ where $\bA_{i,i}$ denotes the $i$th diagonal element, $\bH$ is $n\times n$ random matrix, with i.i.d. entries of zero mean and variance $1/n$. The components of the noise $\bW$ are i.i.d. Gaussian random variables with zero mean and unit variance. The matrix $\bA\bH$ is also known as the \emph{sensing matrix}. We will assume that $\bA$ and $\bH$ are available at the receiver, and that $\bA$ is fixed, namely, given some realization, which determines the number of ones on the diagonal, which will be denoted by $k$. We denote by $q \define k/n$ the sampling rate, or the compression ratio. 

In this paper, we are concerned with the following \emph{mutual information rates} 
\begin{align}
\calI_1\define \limsup_{n\to\infty}\frac{I\p{\bY;\bX\vert\bA,\bH}}{n},
\label{calII1}
\end{align}
and
\begin{align}
\calI_2\define \limsup_{n\to\infty}\frac{I\p{\bY;\bU\vert\bA,\bH,\bS}}{n},
\label{calII2}
\end{align}
which are central in a variety of communications and processing models, see \cite{TulinoCaire,Tulino,Peleg}, and references therein. Usually, $\calI_1$ is evaluated using the \emph{replica method} (see, e.g., \cite{Tulino,GuoShamaiBaron}), while for $\calI_2$ a classical closed-form expression exists \cite{Tulino}. Based on the results in \cite{Wasim2}, we provide an analytic expression for $\calI_1$, which is derived rigorously, and is numerically consistent with the replica predictions. The analytic expressions of $\calI_1$ and $\calI_2$ will lead us to the main objective of this paper, which is to explore the various applications of these quantities in some channel coding problems. 

\section{Mutual Information Rates}
In this subsection, we provide the analytic expressions for $\calI_1$ and $\calI_2$. In the following, we first provide a simple formula for $\calI_1$ which is based on the replica heuristics, and is proved in \cite{Tulino}. For i.i.d. sources, where $f\p{\cdot}$ is linear, we have the following result \cite[Claim 1]{Tulino}.
\begin{claim}[$\calI_1$ via the replica method]\label{th:tul}
Let $B_0,X_0,Z$ be independent random variables, with $B_0\sim \text{Bernoulli-}p$, $X_0\sim\calN\p{0,\sigma^2}$, and $Z\sim\calN\p{0,1}$, and define $V_0 \define B_0X_0$. Then, the limit supremum in \eqref{calII1} is, in fact, an ordinary limit, and
\begin{align}
\calI_1 = I\p{V_0;V_0+\eta^{-1/2}Z} + q\pp{\log\frac{q}{\eta}+\p{\frac{\eta}{q}-1}\log e}
\label{I1replica}
\end{align}
where $\eta$ is the non-negative solution of
\begin{align}
\frac{1}{\eta} = \frac{1}{q}\p{1+\text{mmse}\p{V_0\vert V_0+\eta^{-1/2}Z}}.
\label{fixPointRep}
\end{align}
If the solution of \eqref{fixPointRep} is not unique, then we select the solution that minimizes $\calI_1$ given in \eqref{I1replica}.
\end{claim}

The replica method is not rigorous. Nevertheless, based on a recent paper \cite{Wasim2}, where methods from statistical physics and random matrix theory are used, it is possible to derive $\calI_1$ rigorously. 
Before we state the result, we define some auxiliary functions of a generic variable $x\in\pp{0,1}$:
\begin{align}
&b\p{x} \define \frac{-\pp{1+\sigma^2\p{q-x}}+\sqrt{\pp{1+\sigma^2\p{q-x}}^2+4\sigma^2x}}{2\sigma^2x},\label{firstt}\\
&g\p{x} \define 1+\sigma^2xb\p{x},\\
&\bar{I}\p{x} \define \frac{q}{x}\ln{g\p{x}}-\ln{b\p{x}}-\frac{\sigma^2qb\p{x}}{g\p{x}},\label{firs2tt}\\
&V\p{x}\define\frac{\sigma^4b^2\p{x}x^2}{2g^2\p{x}},\\
&L\p{x}\define\frac{\sigma^2b\p{x}}{2g^2\p{x}},
\end{align}
and
\begin{align}
&t\p{x}\define f\p{x}-\frac{x}{2}\bar{I}\p{x}+V\p{x}\pp{m_aq\sigma^2+q}.\label{lastt}
\end{align}
The mutual information rate $\calI_1$ is given in the following theorem.
\begin{theorem}[$\calI_1$ via the results of \cite{Wasim2}]\label{th:1}
Let $Q$ be a random variable, distributed according to
\begin{align}
\pr_Q\p{w} = \frac{1-m_a}{\sqrt{2\pi P_y}}\exp\p{-\frac{w^2}{2P_y}} +\frac{m_a}{\sqrt{2\pi \p{P_y+q^2\sigma^2}}}\exp\p{-\frac{w^2}{2\p{P_y+q^2\sigma^2}}} 
\end{align}
where $m_a$ is defined as in \eqref{apriorimag} and $P_y \define m_a\sigma^2q+q$. Let us define 
\begin{align}
K\p{Q,\alpha_1,\alpha_2} \define \frac{1}{2}\pp{1+\tanh\p{\frac{L\p{\alpha_1}Q^2-\alpha_2}{2}}}
\label{KtermFluc}
\end{align}
where $\alpha_1\in\pp{0,1}$ and $\alpha_2\in\mathbb{R}$. Let $L'\p{m}$ and $t'\p{m}$ designate the derivatives of $L\p{m}$ and $t\p{m}$ w.r.t. $m$, respectively, and let $m_\circ$ and $\gamma_\circ$ be solutions of the system of equations
\begin{subequations}
\begin{align}
&\gamma_\circ \define-\bE\ppp{K\p{Q,m_\circ,\gamma_\circ}Q^2L'\p{m_\circ}}-t'\p{m_\circ},\label{magnetDet1}\\
&m_\circ \define \bE\ppp{K\p{Q,m_\circ,\gamma_\circ}}.
\label{magnetDet}
\end{align}\label{magnetddd}
\end{subequations}
In case of more than one solution, $\p{m_\circ,\gamma_\circ}$ is the pair with the largest value of
\begin{align}
t\p{m_\circ}+\p{m_\circ-\frac{1}{2}}\gamma_\circ+\bE\ppp{\frac{1}{2}L\p{m_\circ}Q^2+\ln \pp{2\cosh\p{\frac{L\p{m_\circ}Q^2-\gamma_\circ}{2}}}}.
\label{criticalPoint}
\end{align}
Finally, define
\begin{align}
h\p{\gamma_\circ,m_\circ} = \gamma_\circ\p{m_\circ-\frac{1}{2}}+\bE\ppp{\frac{1}{2}L\p{m_\circ}Q^2+\ln \pp{2\cosh\p{\frac{L\p{m_\circ}Q^2-\gamma_\circ}{2}}}}.
\label{hhdef}
\end{align}
Then, the limit supremum in \eqref{calII1} is, in fact, an ordinary limit, and
\begin{align}
\calI_1&=\frac{1}{2}\sigma^2m_aq+\calH_2\p{m_a}+f\p{m_a}-t\p{m_\circ}-h\p{\gamma_\circ,m_\circ}.
\label{DasymMMSE}
\end{align}
\end{theorem}

The proof of Theorem \ref{th:1} is a special case of the one in \cite{Wasim2}, where the asymptotic MMSE was considered. Nonetheless, we provide in Appendix \ref{app:1} a proof outline. Comparing Claim \ref{th:tul} and Theorem \ref{th:1}, it is seen that the results appear to be analytically quite different. Nevertheless, numerical calculations indicate that they are, in fact, equivalent. A representative comparison appears in Fig. \ref{fig:0}.
\begin{figure}[!t]
\begin{minipage}[b]{1.0\linewidth}
  \centering
	\centerline{\includegraphics[width=12cm,height = 9cm]{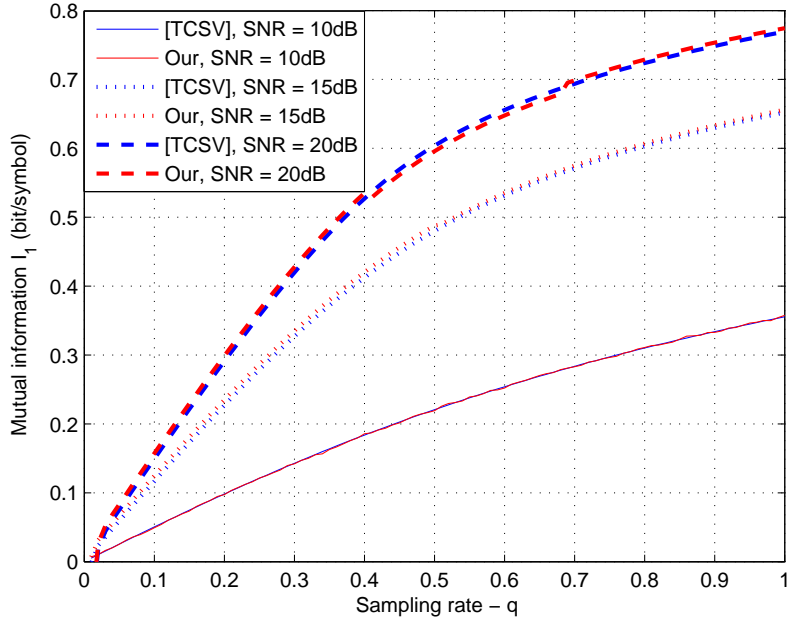}}
	
\end{minipage}
\caption{Mutual information rate $\calI_1$ as a function of the sampling rate $q$, for SNR = 10dB, 15dB, 20dB and $p = 0.2$.}
\label{fig:0}
\end{figure}

Contrary to $\calI_1$, the mutual information rate $\calI_2$ can be fairly easily calculated using, again, random matrix theory. Let
\begin{align}
\calF\p{x,y} \define \p{\sqrt{x\p{1+\sqrt{y}}^2+1}-\sqrt{x\p{1-\sqrt{y}}^2+1}}^2.
\end{align}
The information rate $\calI_2$ is given in the following theorem. 
\begin{theorem}(\cite[Theorem 2]{Tulino})
The information rate $\calI_2$ is given by
\begin{align}
\calI_2 = p\log\pp{1+q\sigma^2-\frac{1}{4}\calF\p{q\sigma^2,\frac{p}{q}}} + q\log\pp{1+p\sigma^2-\frac{1}{4}\calF\p{q\sigma^2,\frac{p}{q}}}-\frac{1}{4\sigma^2}\calF\p{q\sigma^2,\frac{p}{q}}\log e.
\label{calI2}
\end{align}
\end{theorem}

Equipped with closed-from expressions of $\calI_1$ and $\calI_2$, we are now in a position to propose and explore several applications of these information rates.

\section{Channel Coding}\label{sec:coding}
In this section, we consider three different cases that are related to channel coding problems. Generally speaking, the main differences among these cases is in the available knowledge of the transmitter and the receiver about the source. In the following applications, it is assumed that both $\bA$ and $\bH$ are available at the receiver, but are unavailable to the transmitter. Accordingly, the matrix $\bA\bH$ can be considered as part of the channel output, and the mutual information of interest is $I\p{\bY,\bA,\bH;\bX}$. Thus, by using the chain rule of the mutual information and the fact that $\bA$ and $\bH$ are statistically independent of the source $\bX$, we readily obtain that
\begin{align}
I\p{\bY,\bA,\bH;\bX} = I\p{\bY;\bX\vert\bA,\bH},
\label{discsc}
\end{align}
and
\begin{align}
I\p{\bY,\bA,\bH;\bU\vert\bS} = I\p{\bY;\bU\vert\bA,\bH,\bS},
\end{align}
which are simply identified as \eqref{calII1} and \eqref{calII2}, respectively. Keeping these observations in mind, our goal is to provide achievable rates in various channel coding problems, which will only require us to know the mutual information rates $\calI_1$ and $\calI_2$. Finally, note that part of the following coding principles have already appeared in \cite{Peleg}, but relying on bounds.

The input $\bX$ in the previous section was considered as continuous uncoded signal. However, in the following applications, we will deal with coding problems. Accordingly, we use codes and allow the use of the channel \eqref{GeneralModel} for $n$ times as required by the code length. The whole codebook is of size $2^{nR}$ codewords. The transmitter chooses a codeword $\bX$ and transmits it over the channel.

\subsection{Controlled sparsity pattern}\label{subsec:1}
Here, the sparsity pattern $\bS$, as well as the Gaussian signal $\bU$, are assumed to be controlled and given at the transmitter. The constraints are on the average support power, $\sigma^2$, and the sparsity rate, that is the probability $p\define\pr\p{S_i=1}$. One motivation for this setting is, for example, in case where the transmit antennas (conveying $\bX$) are remote, and ``green" communications constraints enforce shutting off a fraction $\p{1-p}$ of the antennas, corresponding to the sparsity of the pattern $\bS$. Here, since the shut-off pattern can be controlled, it can be used to convey information as well. We have the following immediate result.
\begin{theorem}[reliable coding rate]\label{th:app1}
Assume the source-channel statistics assumptions that are given in Section \ref{sec:model}, and assume that $\bS$ and $\bU$ can be controlled by the transmitter. Then, $\calI_1$ in \eqref{DasymMMSE} (or in \eqref{I1replica}) is an achievable information rate for reliable communication.  
\end{theorem}
\begin{proof}
Since both $\bS$ and $\bU$ are controlled, then $\bX$ is also controlled. Note, however, that $\bS$ is not provided to the receiver beforehand. Thus, this is just a channel with inputs $\p{\bS,\bU}$ and output $\bY$, where the matrices $\bH$ and $\bA$ are provided to the receiver only (the transmitter is aware of the statistics of course). Therefore, an achievable coding rate is given by (recall \eqref{discsc})
\begin{align}
\limsup_{n\to\infty}\frac{I\p{\bS,\bU;\bY\vert\bA,\bH}}{n} = \limsup_{n\to\infty}\frac{I\p{\bX;\bY\vert\bA,\bH}}{n},
\end{align}
which is exactly $\calI_1$. 
\end{proof}

Recall that the information rate $\calI_1$, given in Theorem \ref{th:1}, is valid also for sources that are not necessarily memoryless, as we allowed the model given in \eqref{inputassmeas} with a general function $f$. It is then interesting to check whether optimization over this class of sources can help to increase $\calI_1$. Let
\begin{align}
\mathscr{F}\define \ppp{f:\pp{0,1}\to\left(-\infty,0\right],\;f\in \calA\pp{0,1}}
\label{calFdef}
\end{align}
where $\calA\pp{0,1}$ is the class of analytic functions on the interval $\pp{0,1}$. Then, according to \eqref{inputassmeas}, our class of sources is uniquely determined by the set of functions $\mathscr{F}$. Also, let $f_L$ designate the affine function $f_L\p{m} = am+b$, where $a,b\in\mathbb{R}$, and recall that substitution of $f_L$ in the pattern measure \eqref{inputassmeas} corresponds to a memoryless assumption of the sparsity pattern. We have the following result. Finally, let $\mathscr{P}_s$ be the set of probability distributions of the form of \eqref{inputassmeas}.

\begin{theorem}[memoryless pattern is optimal over $\mathscr{P}_s$]\label{th:app2}
Under the \emph{asymptotic average sparseness} constraint, defined as
\begin{align}
\lim_{n\to\infty}\frac{1}{n}\bE\ppp{\sum_{i=1}^nS_i} = p,
\label{theConstr}
\end{align}
the following holds
\begin{align}
\max_{\mathscr{P}_s}\;\calI_1 \equiv \max_{\mathscr{F}}\;\calI_1 = \left.\calI_1\right|_{f = f_L}.
\end{align}
In words, memoryless patterns give the maximum achievable rate over $\mathscr{P}_s$.
\end{theorem}
\begin{proof}
See Appendix \ref{app:2}
\end{proof}
Intuitively speaking, Theorem \ref{th:app2} is essentially expected due to the natural symmetry in our model induced by the assumptions on $\bA$ and $\bH$, that are given only at the receiver side (had these matrices been known to the transmitter, this result may no longer be true). Also, note that when $\bS=\p{1,1,\ldots,1}$, namely, the source is not sparse, we obtain a MIMO setting, in which it is well-known that the Gaussian i.i.d. process achieves capacity \cite{Taleter}. 
 %
In the following, we show that the optimal distribution of the pattern sequence must be invariant to permutations.
\begin{theorem}[permutation invariant distribution]\label{th:app3}
Let $\mathscr{S}$ be the set of all probability distributions of $\bS$, and let $\mathscr{S}_\Pi$ denote the set of all probability distributions that are invariant to permutations. Then,
\begin{align}
\max_{\mathscr{S}}\;\calI_1 = \max_{\mathscr{S}_\Pi}\;\calI_1.
\end{align}
\end{theorem}
\begin{proof}
The maximization of $I\p{\bY;\bX\vert\bA,\bH}$ over $\mathscr{S}$ boils down to the maximization of the conditional entropy $H\p{\bY\vert\bA,\bH}$, namely,
\begin{align}
\arg\max_{\mathscr{S}}\;I\p{\bY;\bX\vert\bA,\bH} &= \arg\max_{\mathscr{S}}\;H\p{\bY\vert\bA,\bH}\\
&= \arg\max_{\mathscr{S}}\;\bE\pp{\log \frac{1}{\pr\p{\bY\vert\bA,\bH}}}.
\end{align}
Recall that
\begin{align}
\pr\p{\bY\vert\bA,\bH} = \int_{\mathbb{R}^n}\mathrm{d}\bx\pr\p{\bx}\pr\p{\bY\vert\bA,\bH,\bx}.
\end{align}
Since the columns of $\bA\bH$ are i.i.d. and $\p{\bA,\bH}$ are known solely to the receiver, it is evident that the conditional entropy $H\p{\bY\vert\bA,\bH}$ is invariant to permutations of $\bS$ in $\pr\p{\bS}$. To see this, let $\pr_\pi\p{\bS}$ denote some permuted version of $\pr\p{\bS}$, namely, $\pr_\pi\p{\bS} = \pr\p{\bpi\bS}$ where $\bpi$ is a permutation matrix corresponding to some permutation. Accordingly, let $\pr_\pi\p{\bX}$ be the probability distribution of $\bX$ induced by the permuted distribution $\pr_\pi\p{\bS}$. Finally, let $H_\pi\p{\bY\vert\bA,\bH}$ designate the conditional entropy of $\bY$ given $\p{\bA,\bH}$ where $\bX$ is distributed according to $\pr_\pi\p{\bX}$. Then,
\begin{align}
H_\pi\p{\bY\vert\bA,\bH} &= -\bE\ppp{\log\int_{\mathbb{R}^n}\mathrm{d}\bx\pr_\pi\p{\bx}\pr\p{\bY\vert\bA,\bH,\bx}}\\
&=-\bE\ppp{\log\int_{\mathbb{R}^n}\mathrm{d}\bx\pr_\pi\p{\bpi\bx}\pr\p{\bY\vert\bA,\bH,\bpi\bx}}\\
&=-\bE\ppp{\log\int_{\mathbb{R}^n}\mathrm{d}\bx\pr\p{\bx}\pr\p{\bY\vert\bA,\bH,\bpi\bx}}
\end{align} 
where in the second equality we changed the variable $\bx\mapsto\bpi\bx$ which permutes the vector $\bx$ according to the permutation used in $\pr_\pi\p{\bS}$. Now,
\begin{align}
&H_\pi\p{\bY\vert\bA,\bH} =-\bE\ppp{\log\int_{\mathbb{R}^n}\frac{1}{\p{2\pi}^{k/2}}\mathrm{d}\bx\pr\p{\bx}\exp\p{-\frac{1}{2}\norm{\bY-\bA\bH\bpi\bx}^2}}\\
& = -\int\mathrm{d}\pr\p{\by\vert\bA,\bH}\mathrm{d}\pr\p{\bA,\bH}\pp{\log\int_{\mathbb{R}^n}\frac{1}{\p{2\pi}^{k/2}}\mathrm{d}\bx\pr\p{\bx}\exp\p{-\frac{1}{2}\norm{\by-\bA\bH\bpi\bx}^2}}\\
& = -\int\mathrm{d}\pr\p{\by\vert\bA,\bH\bpi^T}\mathrm{d}\pr\p{\bA,\bH\bpi^T}\pp{\log\int_{\mathbb{R}^n}\frac{1}{\p{2\pi}^{k/2}}\mathrm{d}\bx\pr\p{\bx}\exp\p{-\frac{1}{2}\norm{\by-\bA(\bH\bpi^T)\bpi\bx}^2}}\\
& = -\int\mathrm{d}\pr\p{\by\vert\bA,\bH}\mathrm{d}\pr\p{\bA,\bH}\pp{\log\int_{\mathbb{R}^n}\frac{1}{\p{2\pi}^{k/2}}\mathrm{d}\bx\pr\p{\bx}\exp\p{-\frac{1}{2}\norm{\by-\bA\bH\bx}^2}}\\
& = H\p{\bY\vert\bA,\bH}
\end{align} 
where in the third equality we changed the variable $\bH\mapsto\bH\bpi^T$, and the forth equality follows from the facts that $\bH\bpi^T\bpi\bx = \bH\bx$ and that $\p{\bA,\bH}$ are i.i.d. and thus $\pr\p{\bA,\bH\bpi^T} = \pr\p{\bA,\bH}$. 

Continuing, let $\pr_*\in\mathscr{S}$ denote the probability distribution that maximize $I\p{\bY;\bX\vert\bA,\bH}$. Let $\Pi_*$ denote the set of probability distributions obtained from $\pr_*$ by all possible permutations of $\bS$, and thus each is achieving the maximal $I\p{\bY;\bX\vert\bA,\bH}$. Also, let
\begin{align}
\pr_{\text{inv}}\p{\bS} \define \frac{1}{\abs{\Pi_*}}\sum_{\pr\in\Pi_*}\pr\p{\bS}.
\label{permutationinv}
\end{align}
Note that $\pr_{\text{inv}}\p{\bS}\in\mathscr{S}_\Pi$, namely, $\pr_{\text{inv}}\p{\bS}$ is invariant to permutations. Finally, let $H\p{\bY\vert\bA,\bH}\vert_{\pr_{\text{inv}}}$ and $H\p{\bY\vert\bA,\bH}\vert_{\pr_*}$  designate the conditional entropies of $\bY$ given $\p{\bA,\bH}$ where $\bS$ is distributed according to $\pr_{\text{inv}}$ and $\pr_*$, respectively. Thus, from the concavity of $H\p{\bY\vert\bA,\bH}$ w.r.t. $\pr\p{\cdot\vert\bA,\bH}$, we have that
\begin{align}
H\p{\bY\vert\bA,\bH}\vert_{\pr_{\text{inv}}} &\define -\bE\ppp{\log\sum_{\bst\in\ppp{0,1}^n}\pr_\pi\p{\bS}\pr\p{\bY\vert\bA,\bH,\bS}}\\
&\geq -\frac{1}{\abs{\Pi_*}}\sum_{\pr\in\Pi_*}\bE\ppp{\log\sum_{\bst\in\ppp{0,1}^n}\pr\p{\bS}\pr\p{\bY\vert\bA,\bH,\bS}}\\
& = H\p{\bY\vert\bA,\bH}\vert_{\pr_*}
\label{concvv}
\end{align}
where \eqref{concvv} follows from the fact that the conditional entropy is the same for all members of $\Pi_*$ as was mentioned previously. 
\end{proof}

It is tempting to tie Theorems \ref{th:app2} and \ref{th:app3} to infer that the optimal distribution of $\bS$ over $\mathscr{S}$ is memoryless. However, there is still a little gap. Indeed, despite the fact that permutation invariant distributions must depend on the pattern only through the magnetization, not every such distribution can be expressed as the one in \eqref{inputassmeas}, due to the smoothness requirement of $f$. For example, in case of uniform distributions over types, the function $f$ is not continuous. Nonetheless, roughly speaking, it is evident that one can approximate arbitrarily closely such non-smooth behaviors by a respectively smooth function $f$. So, we conjecture that the maximum mutual information is indeed achieved by a memoryless source.  

Finally, we present in Fig. \ref{fig:a} the mutual information rate $\calI_1$ as a function of the sampling rate $q$ and the SNR for $p = 0.2$. It can be seen that increase of the rate or/and the SNR results in an increase of $\calI_1$, as one should expect. 
\begin{figure}[!t]
\begin{minipage}[b]{1.0\linewidth}
  \centering
	\centerline{\includegraphics[width=12cm,height = 9cm]{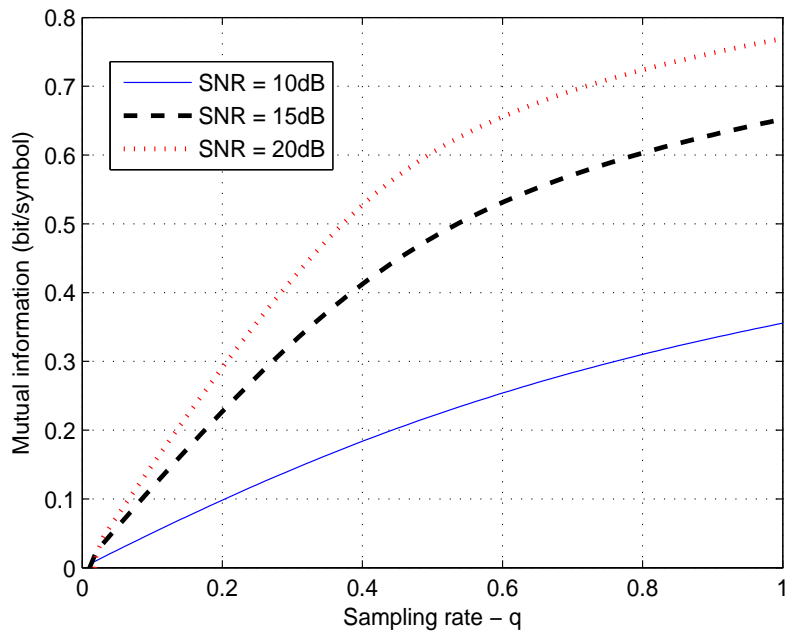}}
	
\end{minipage}
\caption{Mutual information rate $\calI_1$ as a function of $q$ and the SNR for $p = 0.2$.}
\label{fig:a}
\end{figure}

\subsection{Unknown sparsity pattern}\label{Uncontrolled1}
In this subsection, we consider the case where the sparsity pattern is unknown to all parties. The vector $\bU$ is treated as the information to be transmitted over the channel.  
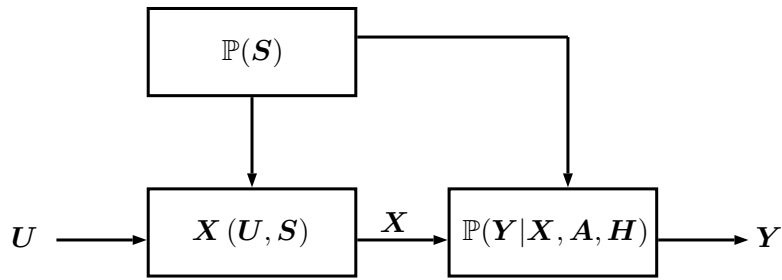
\begin{figure}[!t]
\centering
\begin{pspicture}(0,-1.8)(10.9,1.8)
\psframe[linecolor=black, linewidth=0.04, dimen=outer](5.2,1.8)(2.4,0.6)
\psframe[linecolor=black, linewidth=0.04, dimen=outer](5.2,-0.6)(2.4,-1.8)
\psline[linecolor=black, linewidth=0.04, arrowsize=0.05291666666666668cm 2.0,arrowlength=1.4,arrowinset=0.0]{->}(3.8,0.6)(3.8,-0.2)(3.8,-0.6)
\rput[bl](3.4,1.0){$\pr(\bS)$}
\rput[bl](3.0,-1.4){$\bX\p{\bU,\bS}$}
\psline[linecolor=black, linewidth=0.04, arrowsize=0.05291666666666667cm 2.0,arrowlength=1.4,arrowinset=0.0]{->}(1.2,-1.3)(2.4,-1.3)
\psline[linecolor=black, linewidth=0.04, arrowsize=0.05291666666666667cm 2.0,arrowlength=1.4,arrowinset=0.0]{->}(9.2,-1.3)(10.4,-1.3)
\rput[bl](0.6,-1.4){$\bU$}
\rput[bl](10.5,-1.4){$\bY$}
\psframe[linecolor=black, linewidth=0.04, dimen=outer](9.2,-0.6)(6.4,-1.8)
\psline[linecolor=black, linewidth=0.04, arrowsize=0.05291666666666667cm 2.0,arrowlength=1.4,arrowinset=0.0]{->}(5.2,-1.3)(6.4,-1.3)
\rput[bl](6.6,-1.4){$\pr(\bY\vert\bX,\bA,\bH)$}
\psline[linecolor=black, linewidth=0.04](5.2,1.4)(8.0,1.4)(8.0,1.4)
\psline[linecolor=black, linewidth=0.04, arrowsize=0.05291666666666667cm 2.0,arrowlength=1.4,arrowinset=0.0]{->}(8.0,1.4)(8.0,-0.6)
\rput[bl](5.5,-1.2){$\bX$}
\end{pspicture}
\caption{Gel'fand-Pinsker channel.}
\label{fig:1}
\end{figure}
In this setting, we have the following result.
\begin{theorem}[unknown sparsity pattern]\label{th:app4}
The channel $\pr\p{\cdot\vert\bX,\bA,\bH}$, defined in Section \ref{sec:model}, has an achievable rate given by
\begin{align}
R = \calI_1 - \calH_2\p{m_a}.
\label{app4result}
\end{align}
\end{theorem}
\begin{proof}
This is a channel with input $\bU$ and output $\bY$, where the matrices $\bA$ and $\bH$ are known only to the receiver. Therefore,
\begin{align}
I\p{\bU;\bY\vert\bA,\bH,\bS}&\geq I\p{\bU;\bY\vert\bA,\bH} \\
&= I\p{\bU,\bS;\bY\vert\bA,\bH}-I\p{\bS;\bY\vert\bU,\bA,\bH}\\
&\geq I\p{\bX;\bY\vert\bA,\bH} - H(\bS),
\end{align}
and the result follow, after normalizing by $n$ and taking the limit $n\to\infty$.
\end{proof}

Yet another interesting setting is the case in which the transmitter cannot control the sparsity pattern that is given beforehand. This pattern, $\bS$, is considered to be \emph{channel state} available non-causally/causally to the transmitter solely. The vector $\bU$ is treated as the information to be transmitted over the channel. This framework falls within the well-known Gel'fand-Pinsker channel \cite{Gelfand2} and the Shannon settings \cite{ShannonCaus}, for non-causal and causal knowledge of $\bS$, respectively. This is illustrated in Fig. \ref{fig:1}. A possible motivation for this setting is when the transmitter, that produces the input $\bU$, knows the pattern of switched antennas/shut-off pattern (``green" wireless), but cannot control it. In the following, customary to the Gel'fand-Pinsker and the Shannon settings, the channel state is assumed an i.i.d. process such that $p \define \pr\p{S_i=1}$.

For the case where the side information is available at the transmitter only causally, the capacity expression has been found by Shannon in \cite{ShannonCaus}, and is given by
\begin{align}
\max_{\pr\p{\bvt},\but\p{\bvt,\bst}}\;I\p{\bV;\bY\vert\bA,\bH}
\label{CapShannon}
\end{align}
where $\bU\p{\bV,\bS}$ is a deterministic function of $\bV$ and $\bS$. Note that the auxiliary $\bV$ should be chosen independently of the state \cite{Keshet}, while the transmitted signal can depend on the state. Now, since the sparsity pattern is given, we can adapt the power of the transmitted signal accordingly, that is, we do not transmit at times when $S_i=0$. Accordingly, let us choose $\bV = \bU'$, where $\bU'$ is a Gaussian random vector with independent elements, each with zero mean and variance $p^{-1}\sigma^2$. The transmitted signal is $\bU = \bS\odot\bV$ (which maintains the average power constraint), where $\odot$ denotes the Hadamard product, and thus $\bX = \bS\odot\bU = \bS\odot\bV$, where we have used the fact that $\bS\odot\bS = \bS$. Therefore, \eqref{CapShannon} reads
\begin{align}
I\p{\bV;\bY\vert\bA,\bH} = I\p{\bU';\bY\vert\bA,\bH}.
\end{align}
Unfortunately, we were unable to derive a closed-from expression for the information rate corresponding to $I\p{\bU';\bY\vert\bA,\bH}$. Nonetheless, we note that
\begin{align}
I\p{\bU';\bY\vert\bA,\bH} &=  I\p{\bU',\bS;\bY\vert\bA,\bH}-I\p{\bS;\bY\vert\bU',\bA,\bH} \\
&= I\p{\bX;\bY\vert\bA,\bH}-I\p{\bS;\bY\vert\bA,\bH}\\
&\geq I\p{\bX;\bY\vert\bA,\bH} - H\p{\bS}.
\end{align}
Accordingly, the achievable rate is given by $\calI_{1,S}-\calH_2\p{p}$, where $\calI_{1,S}$ is given in \eqref{I1replica} with $\sigma^2$ replaced by $p^{-1}\sigma^2$, that is the overall SNR is scaled from $p\sigma^2$ to $\sigma^2$. Thus, the improvement due to the knowledge of $\bS$ at the transmitter side compared to Theorem \ref{th:app4} is evident. For the non-causal case, namely, the Gel'fand-Pinsker channel, we could not find a good choice for the auxiliary variable $\bV$. In \cite{GoldsmithDan}, the related case of fading (which may be binary) given as side information known
to the transmitter only was considered. 

Theorems \ref{th:app1} and \ref{th:app4} demonstrate how important it is to be able to control the sparsity pattern $\bS$. Indeed, it can be seen that the gap between these two achievable rates is exactly $\calH_2\p{p}$ which quantifies our uncertainty at the receiver regarding	the source support. This is illustrated in Fig. \ref{fig:b}, which shows the achievable rate as a function of $q$ and the SNR, for $p = 0.2$. It can be seen that there is a significant region of rates and SNR's for which the achievable rate is zero (within this region, the subtractive term in \eqref{app4result} dominates). This is attributed to the fact that the sparsity pattern is uncontrolled, and can be interpreted as the overhead required to the transmitter to adapt to the channel state. 
	%

\begin{figure}[!t]
    \subfloat[\label{fig:b1}]{%
      \includegraphics[width=0.48\textwidth]{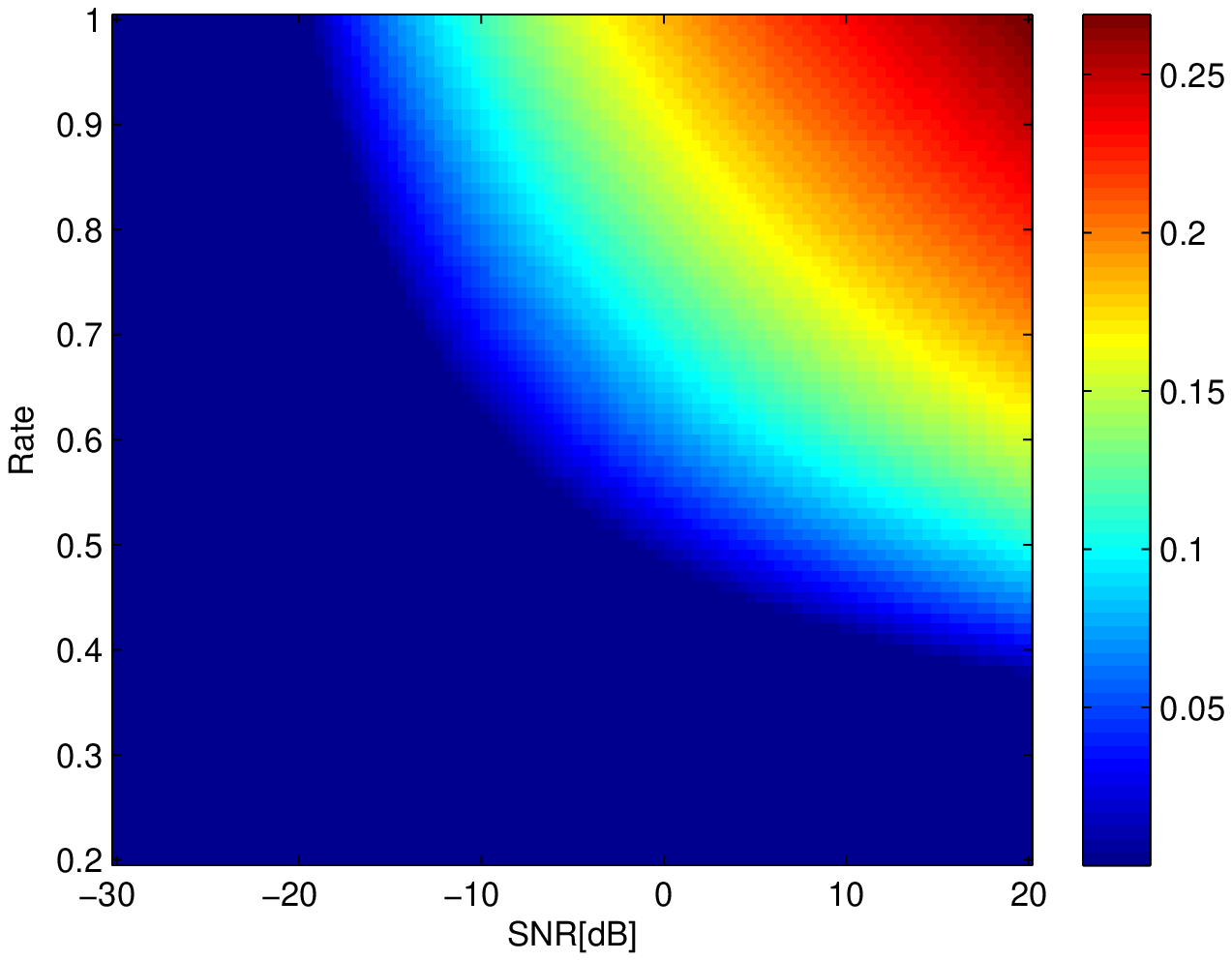}
    }
    \hfill
    \subfloat[\label{fig:b2}]{%
      \includegraphics[width=0.48\textwidth]{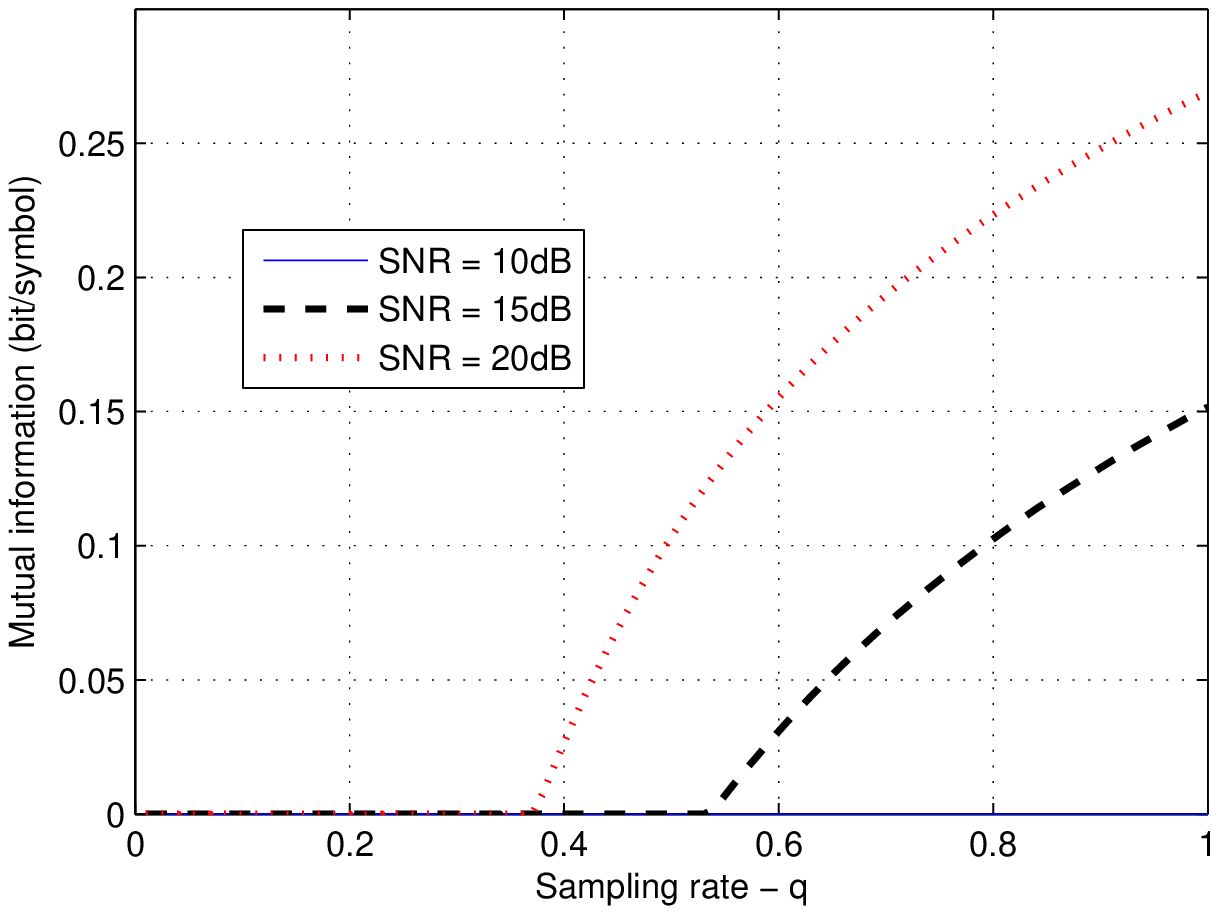}
    }
    \caption{Achievable rate in the uncontrolled sparsity pattern case, as a function of $q$ and the SNR, for $p = 0.2$.}
   \label{fig:b}
\end{figure}

\subsection{The sparsity pattern is carrying the information}\label{SparsityInfo}
In this subsection, we consider the case where the information is conveyed via $\bS$, while $\bU$ plays the role of a fading process, known to nobody. In this case, we have the following result.
\begin{theorem}[informative sparsity pattern]\label{th:app5}
Consider the case in which $\bS$ is carrying the information and $\bU$ is unknown both to the receiver and the transmitter. Then, the achievable rate is given by $R = \calI_1 - \calI_2$.
\end{theorem}
\begin{proof}
Evidently, under the theorem settings, what matters is the mutual information $I\p{\bS;\bY\vert\bA,\bH}$ which readily can be expressed as
\begin{align}
I\p{\bS;\bY\vert\bA,\bH} &= I\p{\bY;\bU,\bS\vert\bA,\bH} - I\p{\bY;\bU\vert\bA,\bH,\bS}\\
&= I\p{\bY;\bU,\bS\vert\bA,\bH} - I\p{\bY;\bU\vert\bA,\bH,\bS}\\
& = I\p{\bY;\bX\vert\bA,\bH} - I\p{\bY;\bU\vert\bA,\bH,\bS},
\label{informationCarSpar}
\end{align}
and thus Theorem \ref{th:app5} follows, after normalizing by $n$ and taking the limit $n\to\infty$.
\end{proof}

Note that similarly to Subsection \ref{subsec:1}, an optimization over the input distribution can be considered. Nonetheless, by using the same arguments it can be shown that there is no gain by using sources with memory. In the following, we consider the high SNR regime. It is not difficult to show that for large $\sigma^2$, the behavior of $\calI_2$ is as follows \cite[Eq. (34)]{Tulino}
\begin{align}
\calI_2 = \min\ppp{q,p}\log\p{1+4\min\ppp{q,p}\sigma^2} + \calO\p{1}
\end{align}
Note that the prelog constant (a.k.a. the degree of freedom) in the above term of $\calI_2$ is just the asymptotic almost-sure rank of the matrix $\bA\bH\bS$, as one should expect. Similarly, the prelog of $\calI_1$ is also $\min\ppp{q,p}$. 
Thus, if we let 
\begin{align}
\calI \define \lim_{n\to\infty}\frac{I\p{\bS;\bY\vert\bA,\bH}}{n},
\end{align}
then following the last observations regarding the prelogs of $\calI_1$ and $\calI_2$, it can be seen that the information rate $\calI$ converges in the high SNR regime to a finite value that is independent of $\sigma^2$. This is not surprising due to the obvious fact that $\calI\leq \calH_2\p{p}$. Fig. \ref{fig:c} shows the achievable rate for $p=0.2$. It is evident that due to the fading induced by $\bU$, there is a significant decrease in the achievable rate.  
\begin{figure}[!t]
\begin{minipage}[b]{1.0\linewidth}
  \centering
	\centerline{\includegraphics[width=12cm,height = 9cm]{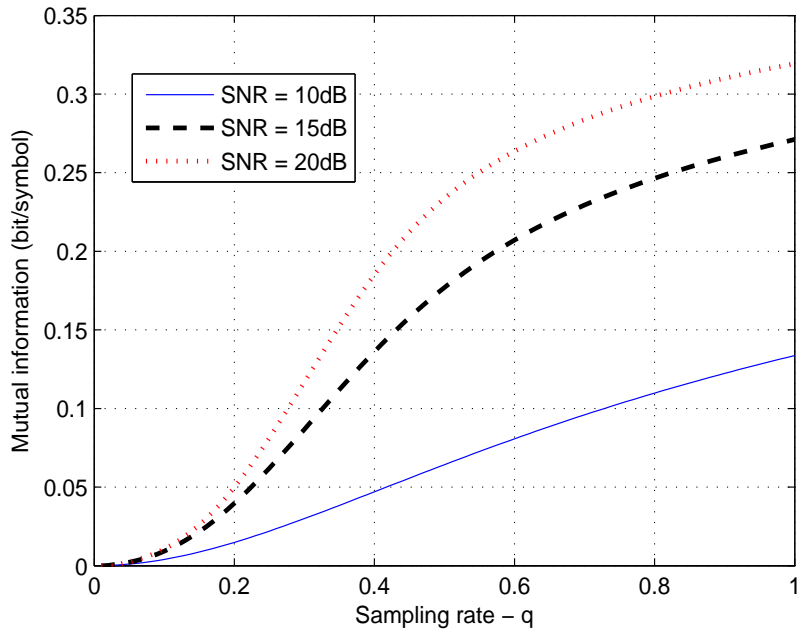}}
	
\end{minipage}
\caption{Achievable rate when the sparsity pattern is carrying the information, as a function of $q$ and the SNR, for $p = 0.2$.}
\label{fig:c}
\end{figure}

\section{The Wiretap Channel}\label{sec:wire}

In the wiretap channel \cite{Wyner}, symbols that are transmitted through a main channel to a legitimate receiver are observed by an eavesdropper across a wiretap channel. The goal of coding for wiretap channels is to facilitate error-free decoding across the main channel, while ensuring that the information transfer rate across the wiretap channel would be as small as possible. A desirable property here is \emph{weak secrecy}, which means that the normalized mutual information between the source and the wiretap channel output will tend to zero.

In our problem, we consider the case in which the legitimate user receives
\begin{align}
\bY_1 = \bA_1\bH_1\bX + \bW_1,
\end{align}
while the eavesdropper receives
\begin{align}
\bY_2 = \bA_2\bH_2\bX + \bW_2.
\end{align}
We assume that the statistics of $\bH_1$ and $\bH_2$ are the same, namely, both are random matrices with i.i.d. elements having variance $1/n$. So is the case for the Gaussian noises $\bW_1$ and $\bW_2$. The difference is, however, between the matrices $\bA_1$ and $\bA_2$, where for $\bA_1$ we define $q_1 \define \pr\p{\bA_{i,i}^{\p{1}}=1}$, for $\bA_2$ we define $q_2 \define \pr\p{\bA_{i,i}^{\p{2}}=1}$, and it is assumed that $q_1\geq q_2$. The motivation could be processing limitations, that is the legitimate receiver has stronger processors, and hence can process more outputs/measurements, going via different jamming patterns, as well as cloud processing (that is the legitimate receiver gets controlled access to more outputs, than the non-legitimate one which has to collect these by chance). 

In a fashion similar to the previous section, we consider here two different cases: Controlled or uncontrolled sparsity pattern (by the transmitter), and unavailable a-priori to both the legitimate and the eavesdropper users. Another configuration that can be considered is when the sparsity pattern $\bS$ is available to both the legitimate user and the eavesdropper, which was already studied in \cite{Liang}.

\subsection{Controlled sparsity pattern}
In this subsection, we consider the case where $\bS$ is controlled by the transmitter, but, is unavailable a-priori to both the legitimate user and the eavesdropper. The \emph{secrecy capacity} is the highest achievable rate that allows perfect weak secrecy, or, in other words, maximal equivocation for the wiretapper. Accordingly, as we deal with degraded channels, our setting is just a special case of \cite{Wyner11}, and the secrecy rate is given by
\begin{align}
\lim_{n\to\infty}\frac{1}{n}\pp{I\p{\bY_1;\bX\vert\bA_1,\bH_1} - I\p{\bY_2;\bX\vert\bA_2,\bH_2}}
\end{align}
which involves only $\calI_1$ terms. Thus, we have the following result.
\begin{theorem}[controlled sparsity pattern]\label{th:app7}
Assume that $\bS$ is controlled by the transmitter, but is available a-priori to neither the legitimate user nor the eavesdropper. Then, the achievable secrecy rate is given by $R = \calI_{1,L} - \calI_{1,E}$, where $\calI_{1,L}$ and $\calI_{1,E}$ are the information rates of the legitimate user and the eavesdropper, given in \eqref{I1replica}, with $q$ replaced by $q_1$ and $q_2$, respectively.
\end{theorem}

Note that similarly to the discussion in Subsection \ref{subsec:1}, one can consider an optimization of the above achievable rate over the class of sources defined in \eqref{inputassmeas}, namely, exploiting the fact that $\bS$ does not have to be Bernoulli. However, by repeating the same steps as in Theorem \ref{th:app2}, it can be shown that there is no gain by using any other source pattern other than the Bernoulli one. 

\begin{theorem}[memoryless pattern is optimal over $\mathscr{P}_s$]\label{th:app10}
Let $\mathscr{F}$ be defined as in \eqref{calFdef}, and let $\mathscr{P}_s$ be the set of probability measures in the form of \eqref{inputassmeas}. Then, under the asymptotic average sparsity constraint, namely,
\begin{align}
\lim_{n\to\infty}\frac{1}{n}\bE\ppp{\sum_{i=1}^nS_i} = p,
\label{theConstr2}
\end{align}
the following holds
\begin{align}
\max_{\mathscr{P}_s}\;\ppp{\calI_{1,L} - \calI_{1,E}} = \max_{\mathscr{F}}\;\ppp{\calI_{1,L} - \calI_{1,E}} = \left.\ppp{\calI_{1,L} - \calI_{1,E}}\right|_{f = f_L}.
\end{align}
In words, memoryless patterns give the maximum achievable rate over $\mathscr{P}_s$.
\end{theorem}
\begin{proof}
See Appendix \ref{app:3}.
\end{proof}
Again, this result is expected due to the symmetry of the assumed model, and the fact that $\bA$ and $\bH$ are available only at the receivers side. Had these matrices been known also to the transmitter, then by controlling the sparsity pattern better secrecy is expected. Finally, similarly to the discussion in Subsection \ref{SparsityInfo}, in the high SNR regime, it is evident that for $q_1\geq q_2\geq p$ the achievable secrecy rate is converges in the high SNR regime to a ﬁnite value that is independent of the SNR. However, if $q_1\geq p>q_2$, then the secrecy rate grows without bound with $\sigma^2$ with prelog constant given by $\p{p-q_2}$. 

Fig. \ref{fig:d} shows the secrecy rate as a function of $q_1$ and the SNR for $p = 0.2$ and $q_2 = 0.3$. It can be seen that when $q_1 = q_2$ the secrecy rate vanishes, as one should expect. Also, for any $q_1>0.3$, increasing the SNR resulting in an increasing of the secrecy rate, and similarly stronger legitimate receivers can achieve higher secrecy rate.

\begin{figure}[!t]
\begin{minipage}[b]{1.0\linewidth}
  \centering
	\centerline{\includegraphics[width=12cm,height = 9cm]{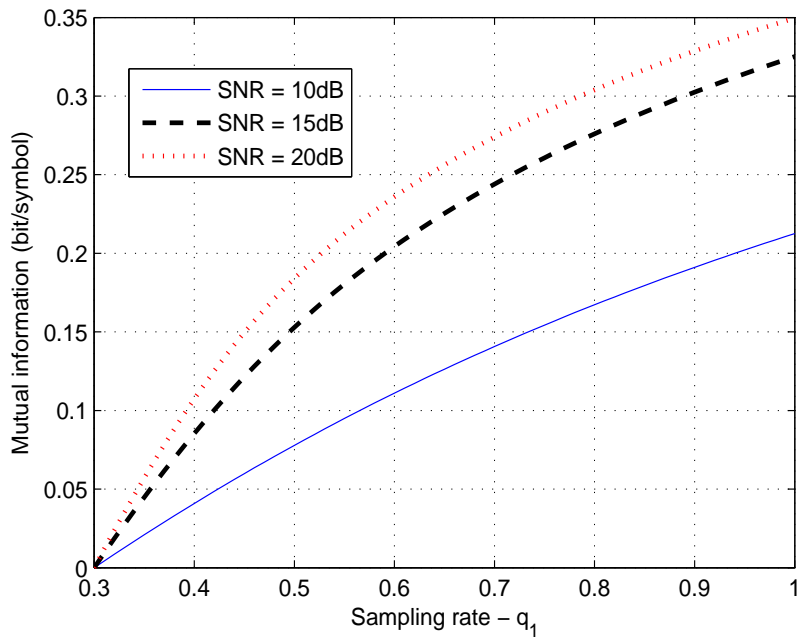}}
	
\end{minipage}
\caption{Secrecy rate when the sparsity pattern is controlled, as a function of $q_1$ and the SNR, for $p = 0.2$ and $q_2 = 0.3$.}
\label{fig:d}
\end{figure}

\subsection{Unavailable sparsity pattern}
In this subsection, we consider the case where the sparsity pattern is known to nobody, and the vector $\bU$ is treated as the information to be transmitted over the channel. As before, since we deal with degraded channels, our setting is just a special case of \cite{Wyner11}, and the secrecy rate is now given by
\begin{align}
\lim_{n\to\infty}\frac{1}{n}\pp{I\p{\bY_1;\bU\vert\bA_1,\bH_1} - I\p{\bY_2;\bU\vert\bA_2,\bH_2}}
\label{initRegion}
\end{align}
Thus, we have the following result.
\begin{theorem}[unavailable sparsity pattern]\label{th:app7a}
Assume that $\bS$ is known to nobody. Then, an achievable secrecy rate is given by 
\begin{align}
\calI_{1,L}-\calI_{2,E}-\calH_2\p{p}
\label{UnCapReg}
\end{align}
\end{theorem}
\begin{proof}
Using \eqref{initRegion}, we note that
\begin{align}
I\p{\bY_1;\bU\vert\bA_1,\bH_1} - I\p{\bY_2;\bU\vert\bA_2,\bH_2} &\stackrel{(a)}{=} I\p{\bX;\bY_1\vert\bA_1,\bH_1} - I\p{\bS;\bY_1\vert\bU,\bA_1,\bH_1} \nonumber\\
& \ \ \ -I\p{\bX;\bY_2\vert\bA_2,\bH_2}+I\p{\bS;\bY_2\vert\bU,\bA_2,\bH_2}\\
& \stackrel{(b)}{\geq}  I\p{\bX;\bY_1\vert\bA_1,\bH_1} - H\p{\bS} \nonumber\\
&\ \ \ - I\p{\bX;\bY_2\vert\bA_2,\bH_2}+I\p{\bS;\bY_2\vert\bA_2,\bH_2}\\
&\stackrel{(c)}{\geq}I\p{\bX;\bY_1\vert\bA_1,\bH_1} - H\p{\bS} - I\p{\bU;\bY_2\vert\bA_2,\bH_2,\bS}\label{last}
\end{align}
where $(a)$ follows from the chain rule of the mutual information, $(b)$ follows from the fact that $I\p{\bS;\bY_2\vert\bU,\bA_2,\bH_2}\geq I\p{\bS;\bY_2\vert\bA_2,\bH_2}$, which in turn is due to
\begin{align}
I\p{\bS;\bY_2\vert\bA_2,\bH_2}&\leq I\p{\bS;\bY_2,\bU\vert\bA_2,\bH_2}\\
& = I\p{\bS;\bU\vert\bA_2,\bH_2} + I\p{\bS;\bY_2\vert\bU\bA_2,\bH_2}\\
& = I\p{\bS;\bY_2\vert\bU,\bA_2,\bH_2}
\end{align}
where the first passage is due to the data processing inequality. Finally, $(b)$ follows from \eqref{informationCarSpar}. Therefore, \eqref{UnCapReg} readily follows from \eqref{last}.
\end{proof}

Fig. \ref{fig:d2} shows the secrecy rate as a function of $q_1$ for $p = 0.2$, various values of the SNR, and $q_2 = 0.1$ and $q_2=0.2$. The results illustrate, again, the importance of controlling the sparsity pattern. 

\begin{figure}[!t]
\begin{minipage}[b]{1.0\linewidth}
  \centering
	\centerline{\includegraphics[width=12cm,height = 9cm]{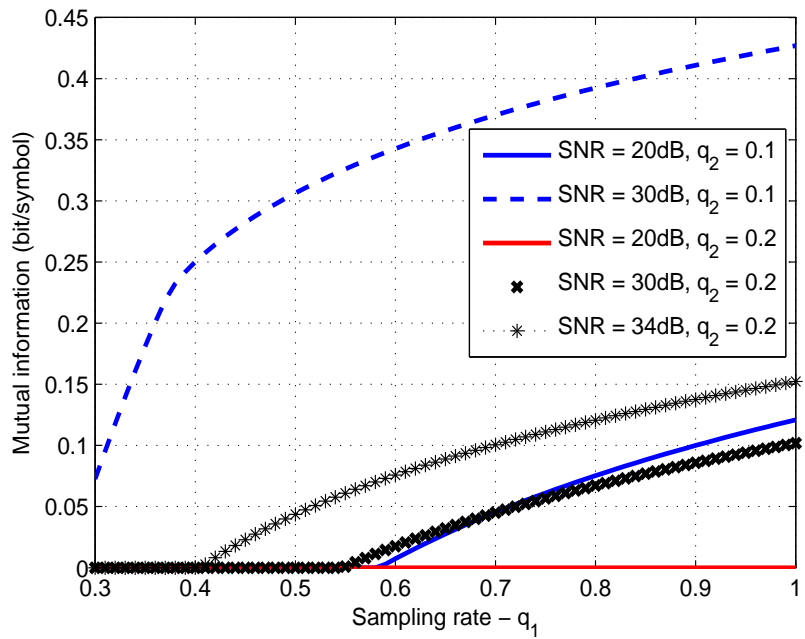}}
	
\end{minipage}
\caption{Secrecy rate when the sparsity pattern is unavailable, as a function of $q_1$ and the SNR, for $p = 0.2$ and $q_2 = 0.3$.}
\label{fig:d2}
\end{figure}

\subsection{Uncontrolled sparsity pattern}
Finally, we consider the case in which $\bS$ is non-causally available to the transmitter, but cannot be controlled, that is, $\bS$ plays the role of a state as in Subsection \ref{Uncontrolled1}. The problem of secrecy capacity here, is not fully solved, but an insightful achievable region was found in \cite{ChenVinck}. This achievable rate is given by
\begin{align}
\lim_{n\to\infty}\frac{1}{n}\pp{I\p{\bV;\bY_1\vert\bA_1,\bH_1} - \max\ppp{I\p{\bV;\bS},I\p{\bV;\bY_2\vert\bA_2,\bH_2}}}
\end{align}
where $\bV-\p{\bU,\bS}-\p{\bY_1,\bY_2}$. Note that, as before, $\bY_2$ can be represented as a degraded version of $\bY_1$. Evidently, this achievable rate is again composed of $\calI_1$ terms, as well as $I\p{\bV;\bS}$. Taking $\bV = \bS\bU$, we obtain the following result.
\begin{theorem}[uncontrolled sparsity pattern]\label{th:app9}
Assume that $\bS$ is a non-causal state information, that is unavailable a-priori to both the legitimate user and the eavesdropper. Then, the achievable secrecy rate is given by 
\begin{align}
R = \calI_{1,L}-\max\ppp{\calH_2\p{p},\calI_{1,E}}.
\end{align}
\end{theorem}

Theorems \ref{th:app7} and \ref{th:app9} demonstrate some gain that results from the ability to control the sparsity pattern control the sparsity pattern $\bS$. Indeed, it can be seen that for high SNR there is no difference between the two achievable secrecy rates. However, below some SNR level, when the sparsity pattern cannot be controlled, the binary entropy $\calH_2\p{p}$ dominates $\calI_{1,E}$, and the resulting secrecy rate is smaller than the secrecy rate in case of controlled sparsity pattern. 

Fig. \ref{fig:e} shows the achievable rate as a function of $q_1$ and the SNR, for $p = 0.2$ and $q_2 = 0.3$. It can be seen that the result is similar to Fig. \ref{fig:c}, that is
\begin{align}
\calI_{1,L}-\max\ppp{\calH_2\p{p},\calI_{1,E}} = \calI_{1,L} - \calH_2\p{p}.
\end{align}
Accordingly, this means that under the above specific choice of $p$ and $q_2$, the loss in the secrecy rate is attributed more to the fact that the sparsity pattern cannot be controlled, than due to the presence of a wiretapper. 
\begin{figure}[!t]
\begin{minipage}[b]{1.0\linewidth}
  \centering
	\centerline{\includegraphics[width=12cm,height = 9cm]{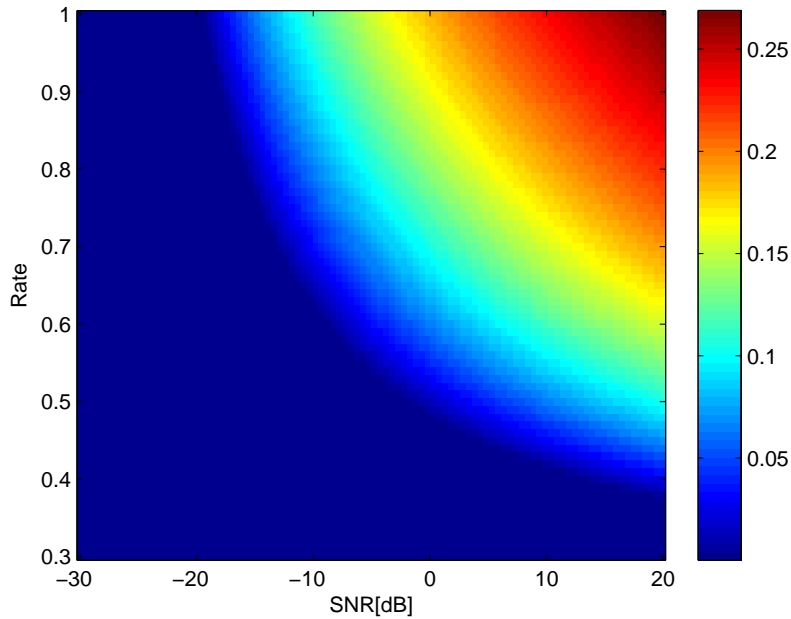}}
	
\end{minipage}
\caption{Secrecy rate in case of an uncontrolled sparsity pattern, as a function of $q_1$ and the SNR, for $p = 0.2$ and $q_2 = 0.3$.}
\label{fig:e}
\end{figure}
In order to illustrate the loss due to the wiretapper, we consider the following example. Figures \ref{fig:f} and \ref{fig:g} show, respectively, the achievable rate and $\calI_{1,L} - \calH_2\p{p}$, as a function of $q_1$ and the SNR, for $p = 0.2$ and $q_2 = 0.5$. In this case the eavesdropper has a strong processor, so it can process more measurements compared to the previous example. Accordingly, it is evident that in this case the wiretapper plays a role, and the loss in the secrecy rate is now more significant.  

\begin{figure}[!ht]
    \subfloat[\label{fig:f}]{%
      \includegraphics[width=0.47\textwidth]{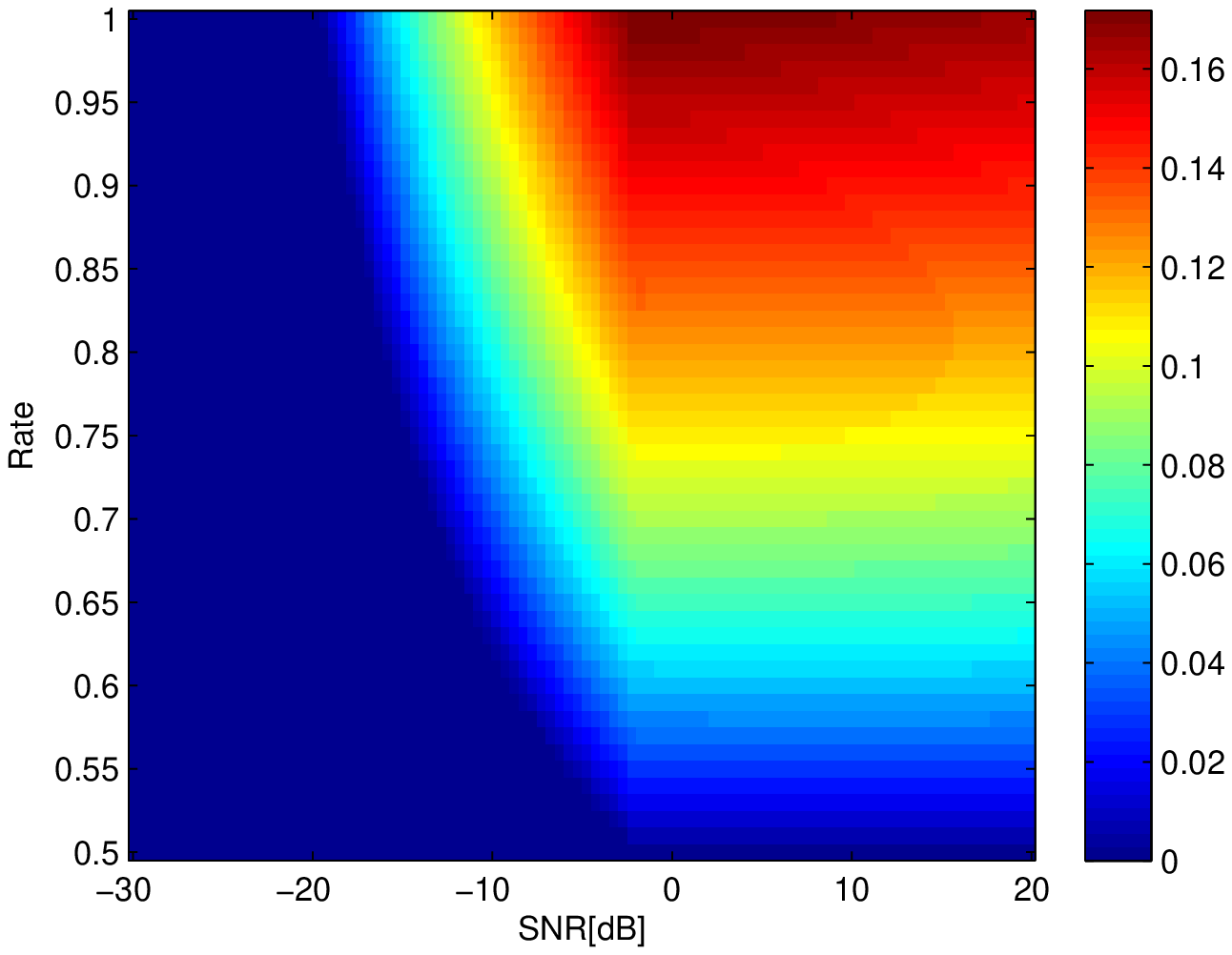}
    }
    \hfill
    \subfloat[\label{fig:g}]{%
      \includegraphics[width=0.47\textwidth]{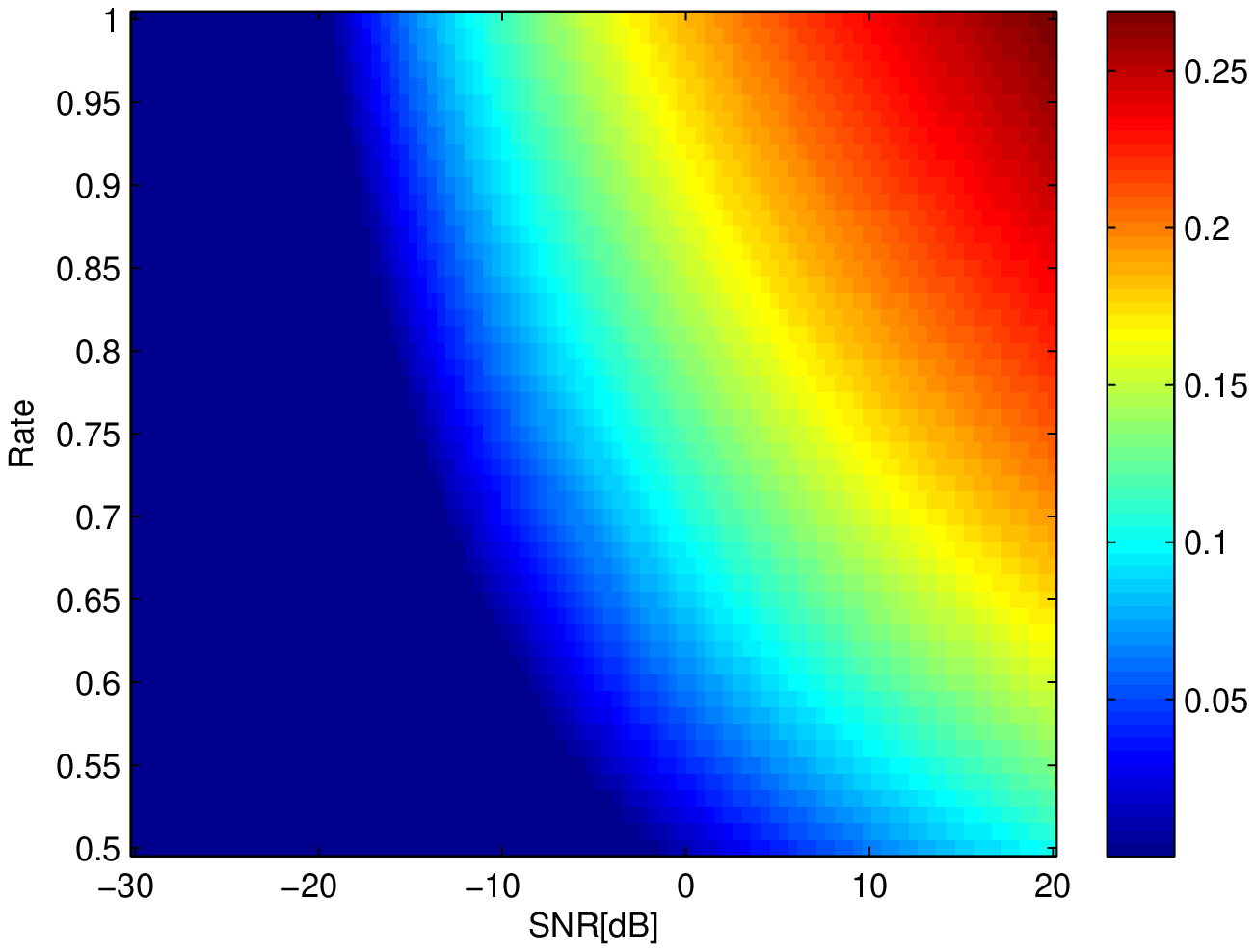}
    }
    \caption{(a) Secrecy rate and (b) $\calI_{1,L}-\calH_2\p{p}$ in case of an uncontrolled sparsity pattern as a function of $q_1$ and the SNR, for $p = 0.2$ and $q_2 = 0.5$.}
   \label{fig:dummy}
\end{figure}

\section{The Multiple Access Channel}\label{sec:MAC}

In this section, we consider the symmetric\footnote{The symmetry is in the sense that all the users transmit at equal power levels.} MAC settings \cite{cover}, in which several senders send information to a common receiver. In our case, we have the following setting: The sequence $\ppp{U_i}$ are now the signals corresponding to different non-cooperative remote users, and the constraint is that on the average, one cannot employ more than $pn$ transmit antennas. The pattern sequence is assumed to be i.i.d. Here, the $i$th user can control the signal $U_i$, as well as $S_i$ (adhering, of course, to the rule that $\pr\p{S_i=1}=p$). We have the following result.
\begin{theorem}[MAC]\label{th:app11}
Consider the MAC under the aforementioned assumptions, and let $\p{R_1,\ldots,R_n}$ denote the rates of the $n$ users. Then,
\begin{align}
\calR_\alpha \leq \p{1-\alpha}^{-1}\calI_{1,\alpha}
\end{align}
where $\calR_\alpha$ is the sum-rates of $n\p{1-\alpha}$ users (no matter which ones, due to symmetry), where $0\leq \alpha< 1$, and $\calI_{1,\alpha}$ equals to $\calI_1$ but with $p$ replaced by $\p{1-\alpha}p$. Particularity, the sum-rates (corresponding to $\alpha=0$) is given by $\calI_1$.
\end{theorem}
\begin{proof}
The case of $\alpha=0$ follows directly from the MAC capacity region \cite{cover}. For the second part, we wish to find the achievable rate of $n\p{1-\alpha}$ users, namely, in the MAC capacity region we condition on the signals produced by the other $n\alpha$ users, and the achievable is given by
\begin{align}
I\p{\bX_{\p{1-\alpha}};\bY\vert\bX_{\alpha},\bA,\bH}
\label{MACAchi}
\end{align}
where $\bX_{\alpha}$ (and similarly for $\bX_{\p{1-\alpha}}$) correspond to the $n\alpha$ users. This can be thought as
\begin{align}
\bY &= \bA\bH\bX + \bW\\
& = \bA\bH\bX_{\p{1-\alpha}} + \bA\bH\bX_{\alpha} +\bW,
\end{align}
and thus \eqref{MACAchi} is equivalent as to examine $\calI_1$ but with $p\mapsto\p{1-\alpha}p$. Finally, due to the fact that $\calI_1$ is normalized by $n$, we need to re-normalize the result by multiplying it by $(1-\alpha)^{-1}$.
\end{proof}

\section{Conclusions}
\label{sec:Conclusion}
In this paper, we examine the problem of sparse sampling of coded signals under several basic channel coding problems. In the first part, we present closed-form single-letter expressions for the input-output mutual information rates, assuming a compressed Gaussian linear channel model. These results are based on rigorous analytical derivations which agree with previously derived results of the replica method. In the second part, we present achievable rates in several channel coding problems, in the wiretap channel model, and in the multiple access channel (MAC). Specifically, for channel coding problem, we consider three cases that differ in the available knowledge of the transmitter and the receiver about the source, and particularity, regarding the sparsity pattern. The results quantify, for example, how important is it to be able to control the sparsity pattern. Also, we show that when this pattern can be controlled by the transmitted, then, a memoryless source maximizes the mutual information rate, given some sparsity average constraint. Then, we consider the wiretap channel model for which several cases were studied. The problems considered are timely and motivated by processing limitations, where the legitimate receiver has stronger processors, and hence can process more outputs/measurements, going via different jamming patterns, as well as cloud processing. Here, the results demonstrate, for example, our inherent limits in achieving some degree of secrecy as a function of the sampling rates of the legitimate user and the eavesdropper. Finally, in a fashion similar to the previous discussion, in case that the sparsity pattern can be controlled by the transmitter, we show that the secrecy rate cannot be increased by using sparsity patterns that are not memoryless. 

\appendices
\numberwithin{equation}{section}
\section{Proof Outline of Theorem \ref{th:1}}
\label{app:1}
In this appendix, we give a proof outline of Theorem \ref{th:1}. It should be emphasized that Theorem \ref{th:1} is a special case of the problem considered in \cite{Wasim2}, and here we emphasize the required modifications. The analysis consists of three main steps, which will be presented in the sequel, along with specific pointers to the proof in \cite{Wasim2}.

The first step in the analysis is to find a generic expression of the mutual information for fixed $k,n$. This is done by using a relationship between the mutual information and some partition function \cite{Neri2}. To this end, we define the following function,
\begin{align}
Z\p{\by,\bH,\bA} \define \int_{\mathbb{R}^n}\mu\p{\mathrm{d}\bx}\exp\pp{-\norm{\by-\bA\bH\bx}^2/2}.
\end{align}
According to our source model assumptions, the input distribution is given by
\begin{align}
\mu\p{\bx} = \sum_{\bst\in\ppp{0,1}^n}\pr\p{\bs}\prod_{i:\;s_i=0}\delta\p{x_i}\prod_{i:\;s_i=1}\frac{1}{\sqrt{2\pi\sigma^2}}e^{-\frac{1}{2\sigma^2}x_i^2}.
\end{align}
Now,
\begin{align}
I\p{\bY;\bX\vert\bA,\bH} &= \bE\ppp{\log\frac{\exp\p{-\norm{\bY-\bA\bH\bX}^2/2}}{Z\p{\by,\bH,\bA}}}\\
&=-\frac{1}{2}\bE\ppp{\norm{\bY-\bA\bH\bX}^2}-\bE\ppp{\log Z\p{\by,\bH,\bA}}\\
&=-\frac{n}{2}-\bE\ppp{\log Z\p{\bY,\bH,\bA}}.
\label{relationMIPART2}
\end{align}
Next, as shown in\footnote{In the notation of \cite{Wasim2}, $\bH$ and $\bH_{\bst}$ correspond to $\bA\bH$ and $\bA\bH_{\bst}$ in our notations.} \cite[Eqs. (57)-(64)]{Wasim2}
\begin{align}
Z\p{\by,\bA,\bH} &= \exp\p{-\frac{1}{2}\norm{\by}^2}\cdot \sum_{\bs\in\ppp{0,1}^n}\pr\p{\bs}\calG\p{\by,\bA,\bHtt}
\end{align}
where
\begin{align}
\calG\p{\by,\bA,\bHtt} \define \frac{\exp\ppp{\frac{1}{2}\by^T\bA\bHtt\calHs\bHtt^T\bA^T\by}}{\sqrt{\det\p{\sigma^2\bHtt^T\bA^T\bA\bHtt+\bI_{\bst}}}},
\label{ShnSte2}
\end{align}
where $\bHtt$ denotes the restriction of $\bH$ on the support $\calS = \ppp{i\in\mathbb{N}: S_i\neq0}$, and $\calHs\define \p{\bHtt^T\bA^T\bA\bHtt+\frac{1}{\sigma^2}\bI_{\bst}}^{-1}$. Thus,
\begin{align}
\frac{I\p{\bY;\bX\vert\bA,\bH}}{n} &= -\frac{1}{2}+\frac{1}{2}\pp{m_a\sigma^2q+1}-\frac{1}{n}\bE\ppp{\log \sum_{\bs\in\ppp{0,1}^n}\pr\p{\bs}\calG\p{\bY,\bA,\bHtt}}\nonumber\\
& = \frac{1}{2}\sigma^2m_aq-\frac{1}{n}\bE\ppp{\log \sum_{\bs\in\ppp{0,1}^n}\pr\p{\bs}\calG\p{\bY,\bA,\bHtt}},
\label{RelationNeri2}
\end{align}
and therefore, in view of \eqref{RelationNeri2}, we wish to calculate the limit
\begin{align}
\lim_{n\to\infty}\frac{1}{n}\bE\ppp{\log\mathscr{Z}_n\p{\bY,\bA,\bH}}\define\lim_{n\to\infty}\frac{1}{n}\bE\ppp{\log \sum_{\bs\in\ppp{0,1}^n}\pr\p{\bs}\calG\p{\by,\bA,\bHtt}}.
\label{qunawasim2}
\end{align}
This concludes the first step. Now, it can be seen from \eqref{ShnSte2} that \eqref{qunawasim2} contains terms that are recognized as an extended version of the Stieltjes and Shannon transforms \cite{Tul2} of the matrix $\bHtt^T\bA^T\bA\bHtt$. In the field of random matrix theory, there is a great interest in exploring the asymptotic behavior, and in particular finding the \emph{deterministic equivalent} of such transforms (see, for example, \cite{baisilbook,coulbook}). Evidently, under some conditions, it is well-known that these transforms asymptotically converge for a fairly wide family of matrices.

Following the last observation, in the second step, we show that these functions converge, with probability tending to one, as $n\to\infty$, to some random functions that are much easier to work with. Accordingly, the following lemma is essentially the core of our analysis; it provides approximations (which are asymptotically exact in the almost sure (a.s.) sense) of $\calG$ and \eqref{qunawasim2}. For simplicity of notations, we let $m_s \define n^{-1}\sum_{i=1}^ns_i$, and recall the auxiliary variables defined in \eqref{firstt}-\eqref{lastt}. The following lemma is proved in \cite[Appendix B, C]{Wasim2}.

\begin{lemma}[asymptotic equivalence]\label{lem:was1} Under the assumptions and definition presented earlier, the following relations hold in the almost sure (a.s.) sense:
\begin{align}
&\lim_{n\to\infty}\frac{1}{n}\ln\det\p{\sigma^2\bHtt^T\bA^T\bA\bHtt+\bItt} = m_s\bar{I}\p{m_s},\label{itm2}
\end{align}
and
\begin{align}
&\lim_{n\to\infty}\frac{1}{n}\pp{\by^T\bA\bHtt\calHs\bHtt^T\bA^T\by-f_{n}} = 0,\label{itm3}
\end{align}
where
\begin{align}
f_{n}\define 2\cdot V\p{m_s}\frac{\norm{\by}^2}{n}+2\cdot L\p{m_s}\frac{\norm{\bHtt^T\bA^T\by}^2}{n}\label{101eq}.
\end{align}
Finally, for large $n$ and $k$, and for $\p{\by,\bA,\bH}$-typical sequences, the function $\mathscr{Z}_n\p{\by,\bA,\bH}$ is lower and upper bounded as follows
\begin{align}
\mathscr{Z}_{-}\p{\by,\bA,\bH}\leq \mathscr{Z}_n\p{\by,\bA,\bH}\leq \mathscr{Z}_{+}\p{\by,\bA,\bH},
\end{align}
where
\begin{align}
&\mathscr{Z}_{\pm}\p{\by,\bA,\bH}\define C_n\cdot\sum_{\bs\in\ppp{0,1}^n}\exp\left\{n\left(\tilde{t}\p{m_s}+L\p{m_s}\frac{1}{n}\sum_{i=1}^n\abs{\by^T\bh_i}^2s_i\pm\varphi\right)\right\},
\label{partplusminus}
\end{align}
in which $C_n$ is the normalization constant in $\pr\p{\bs}$ (see \eqref{inputassmeas}), and
\begin{align}
\tilde{t}\p{m} \define f\p{m}-\frac{m}{2}\bar{I}\p{m}+V\p{m}\frac{\norm{\by}^2}{n},
\end{align}
and the fluctuation term $\varphi$ is typically lower and upper bounded by a vanishing term that is uniform in $\bs$, namely, $\abs{\varphi}\leq\calO\p{1/n}$\footnote{Physically, over the typical set, this fluctuation will not affect the asymptotic behavior of any \emph{intensive} quantity, namely, a quantity that does not depend on $n$ (e.g., the dominant magnetization).}. 
\end{lemma}

The proof of Lemma \ref{lem:was1} is obtained by invoking recent powerful methods from random matrix theory, such as, the Bai-Silverstein method \cite{SilversteinBai}. Equipped with Lemma \ref{lem:was1}, our next and last step is to assess the exponential order of $\mathscr{Z}_{\pm}\p{\by,\bA,\bH}$ using large deviations theory. The following analysis can be found in detail in \cite[Appendix C]{Wasim2}. For completeness, we provide the main ideas here as well. 

First, note that $\mathscr{Z}_{\pm}\p{\by,\bA,\bH}$ can be equivalently rewritten as
\begin{align}
\mathscr{Z}_{\pm}\p{\by,\bA,\bH}  = C_n\cdot\sum_{m_s}\exp\ppp{n\p{\tilde{t}\p{m_s}\pm\varphi}}\hat{\mathscr{Z}}\p{\by,\bA,\bH,m_s}
\label{partplusemin}
\end{align}
where the summation is over $m_s\in\pp{0/n,1/n,\ldots,n/n}$, and
\begin{align}
\hat{\mathscr{Z}}\p{\by,\bA,\bH,m_s}\define\sum_{\bst:\;m\p{\bst} = m_s}\exp\p{L\p{m_s}\sum_{i=1}^n\abs{\by^T\bh_i}^2s_i}
\end{align}
where with slight abuse of notations, the summation is performed over sequences $\bs$ with magnetization, $m\p{\bs} \define n^{-1}\sum_{i=1}^ns_i$, fixed to $m_s$. For the sake of brevity, we will omit the $\pm$ sign. In the following, we will find the asymptotic behavior of $\hat{\mathscr{Z}}\p{\by,\bA,\bH,m_s}$, and then the asymptotic behavior of $\mathscr{Z}_{\pm}\p{\by,\bA,\bH}$. For $\hat{\mathscr{Z}}\p{\by,\bA,\bH,m_s}$, we will need to count the number of sequences $\ppp{\bs}$, having a given magnetization $m_s$, and also admit some linear constraint. Accordingly, consider the following set
\begin{align}
\mathcal{F}_\delta\p{\rho,m}\define\ppp{\bv\in\ppp{0,1}^n:\;\abs{\sum_{i=1}^nv_i-nm}\leq\delta, \;\abs{\sum_{i=1}^nv_iu_{i}-n\rho}\leq\delta}
\label{fcals}
\end{align}
where $\ppp{u_{i}}_{i=1}^n$ is a given sequence of real numbers. Thus, the above set contains binary sequences that admit two linear constraints. We will upper and lower bound the cardinality of $\mathcal{F}_\delta\p{\rho,m}$ for a given $\delta>0$, $m$, and $\rho$. Then, we will use the result in order to approximate $\hat{\mathscr{Z}}\p{\by,\bA,\bH,m_s}$. Using methods that are customary to statistical mechanics, we have the following result which is proved in \cite[Appendix C, eqs. (C.15)-(C.32)]{Wasim2}.
\begin{lemma}
For large $n$ and any $\tau>0$ the cardinality of $\mathcal{F}_\delta\p{\rho,m}$ is upper and lower bounded as follows
\begin{align}
\p{1-\tau}\calV_{-\delta}\leq\abs{\mathcal{F}_\delta\p{\rho,m}}\leq \calV_{\delta}
\end{align}
where
\begin{align}
\log\calV_{\pm\delta}\define \frac{1}{2}\p{\alpha^\circ\sum_{i=1}^nu_i-n\gamma^\circ} - \pp{\alpha^\circ\p{n\rho\mp\delta}-\gamma^\circ\p{nm\mp\delta}}+\sum_{i=1}^n\log\pp{2\cosh\p{\frac{\alpha^\circ u_i-\gamma^\circ}{2}}},
\label{Vdef}
\end{align}
in which $\alpha^\circ,\gamma^\circ$ are given by the solution of the following equations
\begin{align}
\rho = \frac{\delta}{n}+\frac{1}{2n}\sum_{i=1}^nu_{i}+\frac{1}{2n}\sum_{i=1}^n\tanh\p{\frac{\alpha^\circ u_{i}-\gamma^\circ}{2}}u_{i},
\label{saddleeq1}
\end{align}
and
\begin{align}
m = \frac{\delta}{n}+\frac{1}{2}+\frac{1}{2n}\sum_{i=1}^n\tanh\p{\frac{\alpha^\circ u_{i}-\gamma^\circ}{2}}.
\label{saddleeq2}
\end{align}
\end{lemma}

For the purpose of assessing the exponential behavior of $\hat{\mathscr{Z}}\p{\by,\bA,\bH,m_s}$, let us define $u_i = \abs{\by^T\bh_i}^2$. The main observation here is that $\hat{\mathscr{Z}}\p{\by,\bA,\bH,m_s}$ can be represented as 
\begin{align}
\hat{\mathscr{Z}}\p{\by,\bA,\bH,m_s} = 2^n\int_{\calD\subset\mathbb{R}}\exp\p{nL\p{m_s}\rho}\mathscr{C}_n\p{\mathrm{d}\rho}
\label{rhoesin}
\end{align}
where $\calD$ is the codomain\footnote{Note that we do not need to explicitly define $\calD$ simply due to the fact that the exponential term in \eqref{rhoesin} is concave (see \eqref{convexratemax}), and thus the dominating $\rho$ are the same over $\calD$ or over $\mathbb{R}$.} of $\rho$, and $\ppp{\mathscr{C}_n}$ is a sequence of probability measures that are proportional to the number of sequences $\bs$ with $\sum_{i=1}^ns_iu_{i}\approx n\rho$, and $\sum_{i=1}^ns_i\approx nm_s$. These probability measures satisfy the large deviations principle \cite{Dembo,Hollander}, with the following respective lower semi-continuous rate function
\begin{align}
I\p{\rho} = 
\begin{cases}
\log 2-n^{-1}\log \calV_{0}, &\text{if}\;\rho\in\calD\\
\infty, &\text{else}
\end{cases}
\end{align}
where $\calV_{0} \define \lim_{\delta\to0}\calV_\delta$ given in \eqref{Vdef}. Indeed, by definition, the probability measure $\mathscr{C}_n$ is the ratio between $\abs{\mathcal{F}_\delta\p{\rho,m_s}}$ and $2^n$ (the number of possible sequences). Thus, for any Borel set $\calB\subset\calD$, we have that $\lim_{n\to\infty}n^{-1}\log\mathscr{C}_n\p{\calB} = -I\p{\rho}$. Accordingly, due to it large deviations properties, applying Varadhan's theorem \cite{Dembo,Hollander} on \eqref{rhoesin}, one obtains
\begin{align}
\hat{\mathscr{Z}}\p{\by,\bA,\bH,m_s} &\to \exp\pp{n\p{\log 2+L\p{m_s}\rho^\circ-I\p{\rho^\circ}}}
\end{align}
where $\rho^\circ$ is given by (using the fact that the exponential term is convex)
\begin{align}
\rho^\circ& =\arg\max_{\rho\in\mathbb{R}}\ppp{\log 2+L\p{m_s}\rho-I\p{\rho}}\nonumber\\
&=\arg\max_{\rho\in\mathbb{R}}\ppp{L\p{m_s}\rho+n^{-1}\log\calV_0}.
\label{convexratemax}
\end{align}
The maximizer, $\rho^\circ$, is the solution of the following equation
\begin{align}
L\p{m_s}+\frac{1}{n}\frac{\partial}{\partial\rho}\log \calV_0 = 0.
\label{maxrho}
\end{align}
Now, it can be readily shown that (see, \cite[Appendix C, eqs. (C.40)-(C.42)]{Wasim2})
\begin{align}
\frac{1}{n}\frac{\partial}{\partial\rho}\log \calV_0 &= -\alpha^\circ.
\label{llmaxRho}
\end{align}
Thus, using \eqref{llmaxRho} and \eqref{maxrho}, we may conclude that $\alpha^\circ=L\p{m_s}$. Now,
\begin{align}
L\p{m_s}\rho^\circ+\left.n^{-1}\log \calV_0\right|_{\rho^\circ} &=m_s\gamma^\circ+\frac{1}{n}\sum_{i=1}^n\frac{L\p{m_s}u_{i}-\gamma^\circ}{2}+\frac{1}{n}\sum_{i=1}^n\log\pp{2\cosh\p{\frac{L\p{m_s} u_{i}-\gamma^\circ}{2}}}\nonumber\\
&\define \tilde{h}\p{\gamma^\circ,m_s}.
\label{hdefinre}
\end{align}
Therefore, 
\begin{align}
&\hat{\mathscr{Z}}\p{\by,\bA,\bH,m_s} \to\exp\p{n\tilde{h}\p{\gamma^\circ,m_s}}
\label{lassapp}
\end{align}
where $\gamma^\circ$ solves the following equation (see \eqref{saddleeq2})
\begin{align}
m_s = \frac{1}{2n}\sum_{i=1}^n\pp{1+\tanh\p{\frac{L\p{m_s}\abs{\by^T\bh_i}^2-\gamma^\circ}{2}}}.
\end{align}

Thus far, we approximated $\hat{\mathscr{Z}}\p{\by,\bA,\bH,m_s}$. Recalling \eqref{partplusemin}, the next step is to approximate $\mathscr{Z}_{\pm}\p{\by,\bA,\bH}$. Using \eqref{lassapp}, and applying once again Varadhan's theorem (or simply, the Laplace method \cite{NeriMono,Bruijn}) on \eqref{partplusemin}, one obtains that
\begin{align}
{\mathscr{Z}_\pm}\p{\by,\bA,\bH}  &= C_n\cdot\sum_{m_s}\exp\pp{n\p{\tilde{t}\p{m_s}\pm\varphi}}\hat{\mathscr{Z}}\p{\by,\bA,\bH,m_s}\\
&\exe C_n\cdot\exp\ppp{n\p{\tilde{h}\p{\gamma^\circ,m_s^\circ}+\tilde{t}\p{m_s^\circ}\pm\varphi}}
\label{asymZtildeti}
\end{align}
where the dominating $m_s^\circ$ is the saddle point, i.e., one of the solutions to the equation
\begin{align}
\frac{\partial}{\partial m}f\p{m}-\frac{1}{2}\bar{I}\p{m}-\frac{m}{2}\frac{\partial}{\partial m}\bar{I}\p{m}+\frac{\partial}{\partial m}V\p{m}\frac{\norm{\by}^2}{n}+\frac{\partial}{\partial m}\tilde{h}\p{\gamma^\circ,m}=0
\label{sattleh}
\end{align}
where we have used the fact that $\tilde{t}\p{m} = f\p{m}-m\bar{I}\p{m}/2+n^{-1}V\p{m}\norm{\by}^2$. Simple calculations reveal that the derivative of $h\p{\gamma^\circ,m}$ w.r.t. $m$ is given by
\begin{align}
\frac{\partial}{\partial m}\tilde{h}\p{\gamma^\circ,m}= \gamma^\circ+\frac{1}{2n}\sum_{i=1}^n\pp{1+\tanh\p{\frac{L\p{m}\abs{\by^T\bh_i}^2-\gamma^\circ}{2}}}\frac{\partial L\p{m}}{\partial m}\abs{\by^T\bh_i}^2.
\end{align}
Thus, substituting the last result in \eqref{sattleh}, we have that
\begin{align}
\gamma^\circ\p{m_s^\circ} =&-\frac{1}{2n}\sum_{i=1}^n\pp{1+\tanh\p{\frac{L\p{m_s^\circ}\abs{\by^T\bh_i}^2-\gamma^\circ}{2}}}\frac{\partial L\p{m_s^\circ}}{\partial m_s^\circ}\abs{\by^T\bh_i}^2-\frac{\partial}{\partial m_s^\circ}f\p{m_s^\circ}+\frac{1}{2}\bar{I}\p{m_s^\circ}\nonumber\\
&+\frac{m_s^\circ}{2}\frac{\partial}{\partial m_s^\circ}\bar{I}\p{m_s^\circ}-\frac{\partial}{\partial m_s^\circ}V\p{m_s^\circ}\frac{\norm{\by}^2}{n}.
\end{align}
So, hitherto, we obtained that the asymptotic behavior of $\tilde{Z}_{\pm}\p{\by,\bH,\bs}$ is given by \eqref{asymZtildeti}, and the various dominating terms are given by
\begin{subequations}
\begin{align}
&\gamma^\circ\p{m_s^\circ} =-\frac{1}{2n}\sum_{i=1}^n\pp{1+\tanh\p{\frac{L\p{m_s^\circ}\abs{\by^T\bh_i}^2-\gamma^\circ}{2}}}\frac{\partial L\p{m_s^\circ}}{\partial m_s^\circ}\abs{\by^T\bh_i}^2-\frac{\partial}{\partial m_s^\circ}f\p{m_s^\circ}+\frac{1}{2}\bar{I}\p{m_s^\circ}\nonumber\\
&\ \ \ \ \ \ \ \ \ \ +\frac{m_s^\circ}{2}\frac{\partial}{\partial m_s^\circ}\bar{I}\p{m_s^\circ}-\frac{\partial}{\partial m_s^\circ}V\p{m_s^\circ}\frac{\norm{\by}^2}{n},\\
&m_s^\circ = \frac{1}{2n}\sum_{i=1}^n\pp{1+\tanh\p{\frac{L\p{m_s^\circ}\abs{\by^T\bh_i}^2-\gamma^\circ}{2}}}.
\end{align}\label{saddlepoo}
\end{subequations}
Therefore, using \eqref{partplusemin} we obtain
\begin{align}
\lim_{n\to\infty}\frac{1}{n}\log \mathscr{Z}\p{\by,\bA,\bH} =\lim_{n\to\infty}\frac{1}{n}\log C_n +\lim_{n\to\infty}\pp{\tilde{h}\p{\gamma^\circ,m_s^\circ}+\tilde{t}\p{m_s^\circ}}.
\label{lastLimit}
\end{align}
The last thing that is left is to show a concentration property of the saddle point equations given in \eqref{saddlepoo}, and obtain instead the saddle point equations given in \eqref{magnetddd}, which will be also used to assess the limit in \eqref{lastLimit}. Accordingly, we finally obtain that 
\begin{align}
\lim_{n\to\infty}\frac{1}{n}\log \bE\ppp{\mathscr{Z}\p{\by,\bA,\bH}} = \lim_{n\to\infty}\frac{1}{n}\log C_n + h\p{\gamma^\circ,m_s^\circ}+t\p{m_s^\circ}.
\end{align}
This is done by using the theory of convergence of backwards martingale processes, and can be found in \cite[Appendix C, eqs. (C.73)-(C.97)]{Wasim2}. So, eventually, using the relation in \eqref{RelationNeri2}, we finally obtain that
\begin{align}
\lim_{n\to\infty}\frac{1}{n}I\p{\bY;\bX\vert\bA,\bH}&=\frac{1}{2}\sigma^2m_aq-\lim_{n\to\infty}\frac{1}{n}\log C_n-h\p{\gamma^\circ,m_s^\circ}-t\p{m^\circ}\\
&=\frac{1}{2}\sigma^2m_aq+\calH_2\p{m_a}+f\p{m_a}-h\p{\gamma^\circ,m_s^\circ}-t\p{m^\circ},
\end{align}
where in the last equality we have used \eqref{apriorimag} in order to calculate the limit $\lim_{n\to\infty}n^{-1}\log C_n$.
\section{Proof of Theorem \ref{th:app2}}
\label{app:2}
The first equality is obvious. First, by definition (see, \eqref{apriorimag}), $m_a$ is the solution of the following equation
\begin{align}
m_a = \frac{1}{2}\pp{1+\tanh\p{\frac{f'\p{m_a}}{2}}}.
\label{apriorimagnet}
\end{align}
Note that according to \eqref{theConstr}, $m_a = p$. Consider first a polynomial function
\begin{align}
f\p{x} = \sum_{k=1}^M\alpha_k\frac{x^k}{k}
\end{align}
for $x\in\pp{0,1}$, where $M>0$ is natural, and $\ppp{a_l}$ are parameters. Substituting $f$ in \eqref{DasymMMSE}, we see that maximizing $\calI_1$ amounts to maximizing the following function
\begin{align}
\kappa\p{\alpha_1,\ldots,\alpha_M} \define \sum_{k=1}^M\alpha_k\frac{m_a^k}{k} - \sum_{k=1}^M\alpha_k\frac{m_\circ^k}{k} - \tilde{t}\p{m_\circ}-h\p{\gamma_\circ,m_\circ}
\label{kappadef}
\end{align}
where
\begin{align}
\tilde{t}\p{m_\circ} \define t\p{m_\circ} - f\p{m_\circ}.
\label{ttildef}
\end{align}
Now, we take the partial derivative of $\kappa\p{\alpha_1,\ldots,\alpha_M}$ w.r.t. $\alpha_l$ for $1\leq l\leq M$, and readily obtain that
\begin{align}
\frac{\partial}{\partial \alpha_l}\kappa\p{\alpha_1,\ldots,\alpha_M} &= \frac{m_a^l}{l} - \frac{m_\circ^l}{l} - \sum_{k=1}^M\alpha_lm_\circ^{k-1}\frac{\partial m_\circ}{\partial \alpha_l} - \frac{\partial m_\circ}{\partial \alpha_l}\frac{\partial \tilde{t}\p{m_\circ}}{\partial m_\circ} - \frac{\partial h\p{\gamma_\circ,m_\circ}}{\partial \alpha_l}\\
& = \frac{m_a^l}{l} - \frac{m_\circ^l}{l} - \frac{\partial m_\circ}{\partial \alpha_l}\frac{\partial t\p{m_\circ}}{\partial m_\circ} - \frac{\partial h\p{\gamma_\circ,m_\circ}}{\partial \alpha_l}
\label{kappadiv}
\end{align}
where \eqref{kappadiv} follows from \eqref{ttildef}. Using \eqref{hhdef} we obtain
\begin{align}
\frac{\partial h\p{\gamma_\circ,m_\circ}}{\partial \alpha_l} &= \frac{\partial \gamma_\circ}{\partial \alpha_l}\p{m_\circ-\frac{1}{2}} + \gamma_\circ\frac{\partial m_\circ}{\partial \alpha_l} + \bE\ppp{\frac{1}{2}\frac{\partial L\p{m_\circ}}{\partial m_\circ}\frac{\partial m_\circ}{\partial \alpha_l}Q^2}\nonumber\\
&\ +\bE\ppp{\frac{1}{2}\tanh\p{\frac{L\p{m_\circ}Q^2-\gamma_\circ}{2}}\pp{\frac{\partial L\p{m_\circ}}{\partial m_\circ}\frac{\partial m_\circ}{\partial \alpha_l}Q^2-\frac{\partial\gamma_\circ}{\partial \alpha_l}}}\\
& = \gamma_\circ\frac{\partial m_\circ}{\partial \alpha_l}+\bE\ppp{K\p{Q,m_\circ,\gamma_\circ}\frac{\partial L\p{m_\circ}}{\partial m_\circ}\frac{\partial m_\circ}{\partial \alpha_l}Q^2}\label{lastResult}
\end{align}
where the last equality follows from \eqref{magnetDet} and the definition in \eqref{KtermFluc}. Thus, on substituting \eqref{lastResult} in \eqref{kappadiv}, one obtains
\begin{align}
\frac{\partial}{\partial \alpha_l}\kappa\p{\alpha_1,\ldots,\alpha_M} &= \frac{m_a^l}{l} - \frac{m_\circ^l}{l} - \frac{\partial m_\circ}{\partial \alpha_l}\frac{\partial t\p{m_\circ}}{\partial m_\circ} - \gamma_\circ\frac{\partial m_\circ}{\partial \alpha_l} -  \bE\ppp{K\p{Q,m_\circ,\gamma_\circ}\frac{\partial L\p{m_\circ}}{\partial m_\circ}\frac{\partial m_\circ}{\partial \alpha_l}Q^2}\nonumber\\
& = \frac{m_a^l}{l} - \frac{m_\circ^l}{l} - \frac{\partial m_\circ}{\partial \alpha_l}\pp{\gamma_\circ+\frac{\partial t\p{m_\circ}}{\partial m_\circ}+\bE\ppp{K\p{Q,m_\circ,\gamma_\circ}\frac{\partial L\p{m_\circ}}{\partial m_\circ}Q^2}}\nonumber\\
& = \frac{m_a^l}{l} - \frac{m_\circ^l}{l}
\label{kappadiv2}
\end{align}
where the last equality follows from \eqref{magnetDet1}. Setting the above derivatives (for $1\leq l\leq M$) to zero, we see that the stationary sequence of parameters $\ppp{\alpha_k}$ is determined by the solution of the equation
\begin{align}
m_a = m_\circ.
\label{stationarypoint}
\end{align}
To wit, this equation means that the optimal sequence is to be chosen such that the prior and the posterior magnetizations, namely, $m_a$ and $m_\circ$, respectively, be the same. Accordingly, using \eqref{kappadef} and \eqref{stationarypoint}, we obtain that 
\begin{align}
\left.\kappa\p{\alpha_1,\ldots,\alpha_M}\right|_{m_a = m_\circ}=- \tilde{t}\p{m_a}-h\p{\gamma_\circ,m_a},
\label{kappaAA}
\end{align}
which according to the definitions of $m_\circ$, $h\p{\gamma_\circ,m_a}$, and $\tilde{t}\p{m_a}$ given in \eqref{magnetddd}, \eqref{hhdef}, and \eqref{ttildef}, respectively, is a function of $f\p{\cdot}$ (or, equivalently of $\ppp{a_i}$) only through $f'\p{m_a}$. However, by \eqref{apriorimagnet}, we see that the average sparseness constraint fixes the value of $f'\p{m_a}$ to
\begin{align}
f'\p{m_a} = 2\cdot\arctan\p{2m_a-1}.
\label{divdic}
\end{align} 
Therefore, $\left.\kappa\p{\alpha_1,\ldots,\alpha_M}\right|_{m_a = m_\circ}$ given in \eqref{kappaAA} is essentially independent of the specific choice of $\ppp{a_l}$ that admit $m_a = m_\circ$. Now, in terms of $\ppp{\alpha_i}$, the solution to \eqref{stationarypoint} may not be unique. More importantly, there must be a solution corresponding to the memoryless source assumptions, as one can simply fix $\alpha_i=0$ for $2\leq i\leq M$, and then tune $\alpha_1$ such that \eqref{stationarypoint} holds true. Thus, due to the fact that $\calI_1$ is a concave functional w.r.t. $f\p{\cdot}$, we may conclude that this specific choice cannot decrease the maximal value of $\kappa\p{\cdot}$, and hence also that of $\calI_1$. Finally, using standard approximation arguments, since the above derivation is valid for any polynomial, one can approximate any function $f\p{\cdot}$ by using its Taylor series expansion, and obtain the same conclusion.

\section{Proof of Theorem \ref{th:app10}}
\label{app:3}
The first equality is obvious. The second equality is proved exactly in the same way as in the proof of Theorem \ref{th:app2}. Let us start with polynomial $f$ given by
\begin{align}
f\p{x} = \sum_{k=1}^M\alpha_k\frac{x^k}{k}
\end{align}
for $x\in\pp{0,1}$, where $M>0$ is natural, and $\ppp{a_l}$ are parameters. Then, substituting $f$ in \eqref{DasymMMSE}, we see that maximizing $\calI_{1,L} - \calI_{1,E}$ amounts to maximizing the following function (recall that $m_a$ is fixed under the average sparseness constraint)
\begin{align}
\kappa\p{\alpha_1,\ldots,\alpha_M} \define &- \sum_{k=1}^M\alpha_k\frac{m_{\circ,L}^k}{k} - \tilde{t}_L\p{m_{\circ,L}}-h_L\p{\gamma_{\circ,L},m_{\circ,L}} \nonumber\\
&+ \sum_{k=1}^M\alpha_k\frac{m_{\circ,E}^k}{k} - \tilde{t}_E\p{m_{\circ,E}}+h_E\p{\gamma_{\circ,E},m_{\circ,E}}
\label{kappadef2}
\end{align}
where the subscripts ``$L$" and ``$E$" are referring to the legitimate user and the eavesdropper, respectively. For example, $m_{\circ,L}$ and $m_{\circ,E}$ designate the posterior magnetizations of the legitimate and the eavesdropper users, respectively. Also, similarly to the notations used in the proof of Theorem \ref{th:app2}, we define
\begin{align}
\tilde{t}_L\p{m_{\circ,L}} \define t_L\p{m_{\circ_L}} - f\p{m_{\circ,L}},
\label{ttildef2}
\end{align}
and similarly for $\tilde{t}_E\p{m_{\circ,E}}$. Now, we take the partial derivative of $\kappa\p{\alpha_1,\ldots,\alpha_M}$ w.r.t. $\alpha_l$ for $1\leq l\leq M$, and similarly to \eqref{kappadiv}, we obtain that 
\begin{align}
\frac{\partial}{\partial \alpha_l}\kappa\p{\alpha_1,\ldots,\alpha_M} &= -\frac{m_{\circ,L}^l}{l} +\frac{m_{\circ,E}^l}{l}.
\label{kappadiv3}
\end{align}
Setting the above derivatives (for $1\leq l\leq M$) to zero, we see that the stationary sequence of parameters $\ppp{\alpha_k}$ is determined by the solution of the equation
\begin{align}
m_{\circ,L} = m_{\circ,E}.
\label{stationarypoint2}
\end{align}
To wit, this equation means that the optimal sequence is to be chosen such that the posterior magnetizations (of the legitimate user and the eavesdropper) be the same. Accordingly, using the last result and \eqref{kappadef}, we obtain that 
\begin{align}
\left.\kappa\p{\alpha_1,\ldots,\alpha_M}\right|_{m_{\circ,L} = m_{\circ,E}}= - \tilde{t}_L\p{m_{\circ,L}}-h_L\p{\gamma_{\circ,L},m_{\circ,L}} + \tilde{t}_E\p{m_{\circ,L}}+h_E\p{\gamma_{\circ,E},m_{\circ,L}},
\label{divdiff}
\end{align}
which according to the definitions of the various quantities in \eqref{divdiff} depends on $f$ (or, equivalently of $\ppp{a_i}$) only through its derivative $f'\p{m_{\circ,L}}$ (or, equivalently $f'\p{m_{\circ,E}}$). However, equation \eqref{stationarypoint2} essentially fixes the value of $f'\p{m_{\circ,L}}$, and thus $\left.\kappa\right|_{m_{\circ,L} = m_{\circ,E}}$ is independent of the specific choice of source parameters $\ppp{a_l}$ that admit $m_{\circ,L} = m_{\circ,E}$. Whence, using exactly the same arguments as in the proof of Theorem \ref{th:app2}, we conclude that the memoryless choice cannot decrease the maximal value of $\kappa\p{\cdot}$, and hence also that of $\calI_{1,L}-\calI_{1,E}$. 

\ifCLASSOPTIONcaptionsoff
  \newpage
\fi
\bibliographystyle{IEEEtran}
\bibliography{strings}

\begin{thebibliography}{10}
\providecommand{\url}[1]{#1}
\csname url@samestyle\endcsname
\providecommand{\newblock}{\relax}
\providecommand{\bibinfo}[2]{#2}
\providecommand{\BIBentrySTDinterwordspacing}{\spaceskip=0pt\relax}
\providecommand{\BIBentryALTinterwordstretchfactor}{4}
\providecommand{\BIBentryALTinterwordspacing}{\spaceskip=\fontdimen2\font plus
\BIBentryALTinterwordstretchfactor\fontdimen3\font minus
  \fontdimen4\font\relax}
\providecommand{\BIBforeignlanguage}[2]{{%
\expandafter\ifx\csname l@#1\endcsname\relax
\typeout{** WARNING: IEEEtran.bst: No hyphenation pattern has been}%
\typeout{** loaded for the language `#1'. Using the pattern for}%
\typeout{** the default language instead.}%
\else
\language=\csname l@#1\endcsname
\fi
#2}}
\providecommand{\BIBdecl}{\relax}
\BIBdecl

\bibitem{Wasim2}
\BIBentryALTinterwordspacing
W.~Huleihel and N.~Merhav, ``Asymptotic {MMSE} analysis under sparse
  representation modeling,'' \emph{submitted to IEEE Trans. Inf. Theory}, Dec.
  2013. [Online]. Available: \url{http://arxiv.org/abs/1312.3417}
\BIBentrySTDinterwordspacing

\bibitem{Tao}
E.~Cand\'es, J.~Romberg, and T.~Tao, ``Robust uncertainty principles: {E}xact
  signal reconstruction from highly incomplete frequency information,''
  \emph{IEEE Trans. Inf. Theory}, vol.~52, no.~2, pp. 489--509, Feb. 2006.

\bibitem{Donho1}
D.~L. Donoho, ``Compressed sensing,'' \emph{IEEE Trans. Inf. Theory}, vol.~52,
  no.~4, pp. 1289--1306, Apr. 2006.

\bibitem{WuVerdu}
Y.~Wu and S.~Verd\'u, ``Optimal phase transitions in compressed sensing,''
  \emph{IEEE Trans. Inf. Theory}, vol.~58, no.~10, pp. 6241--6263, Oct. 2012.

\bibitem{Gastpar1}
G.~Reeves and M.~Gastpar, ``The sampling rate-distortion tradeoff for sparsity
  pattern recovery in compressed sensing,'' \emph{IEEE Trans. Inf. Theory},
  vol.~58, no.~5, pp. 3065--3092, May 2012.

\bibitem{Tulino}
A.~Tulino, G.~Caire, S.~Verd\'u, and S.~Shamai~(Shitz), ``Support recovery with
  sparsely sampled free random matrices,'' \emph{IEEE Trans. Inf. Theory},
  vol.~59, no.~7, pp. 4243--4271, July 2013.

\bibitem{cc6}
S.~S. Chen, D.~L. Donoho, and M.~A. Saundres, ``Atomic decomposition by basis
  pursuit,'' \emph{SIAM Journal on Scientific Computing}, vol.~20, no.~1, pp.
  33--61, 1999.

\bibitem{cc7}
R.~Tibshirani, ``Regression shrinkage and selection via the lasso,''
  \emph{Journal of the Royal Statistical Society, Series B}, vol.~58, no.~1,
  pp. 267--288, 1996.

\bibitem{cc8}
D.~L. Donoho, A.~Maleki, and A.~Montanari, ``Message-passing algorithms for
  compressed sensing,'' in \emph{Proceedings of the National Academy of
  Sciences}, vol. 106, Nov. 2009, pp. 18\,914--18\,919.

\bibitem{GuoShamaiBaron}
D.~Guo, D.~Baron, and S.~Shamai~(Shitz), ``A single-letter characterization of
  optimal noisy compressed sensing,'' in \emph{Forty-Seventh Annual Allerton
  Conference on Communication, Control, and Computing}.\hskip 1em plus 0.5em
  minus 0.4em\relax Allerton Retreat Center, Monticello, Illinois, Sep. 30-Oct.
  2, 2009.

\bibitem{Neri1}
N.~Merhav, ``Optimum estimation via gradients of partition functions and
  information measures: A statistical-mechanical perspective,'' \emph{IEEE
  Trans. Inf. Theory}, vol.~57, no.~6, pp. 3887--3898, June 2011.

\bibitem{baisilbook}
Z.~Bai and J.~W. Silverstein, \emph{Spectral Analysis of Large Dimensional
  Random Matrices}.\hskip 1em plus 0.5em minus 0.4em\relax Springer, 2010.

\bibitem{coulbook}
R.~Couillet and M.~Debbah, \emph{Random Matrix Methods for Wireless
  Communications}.\hskip 1em plus 0.5em minus 0.4em\relax Cambridge University
  Press, 2011.

\bibitem{TulinoCaire}
A.~Tulino, G.~Caire, S.~Shamai, and S.~Verd\'u, ``Capacity of channels with
  frequency-selective and time-selective fading,'' \emph{IEEE Trans. on Inf.
  Theory}, vol.~56, no.~3, pp. 1187--1215, Mar. 2010.

\bibitem{Peleg}
M.~Peleg and S.~Shamai, ``On sparse sensing and sparse sampling of coded
  signals at sub-landau rates,'' \emph{Transactions on Emerging
  Telecommunications Technologies}, Dec. 2013.

\bibitem{Abbas}
A.~El~Gamal and Y.~H. Kim, \emph{Network Information Theory}.\hskip 1em plus
  0.5em minus 0.4em\relax Cambridge University Press, 2012.

\bibitem{ShannonCaus}
C.~Shannon, ``Channels with side information at the transmitter,'' \emph{IBM J.
  Res. and Dev.}, vol.~2, no.~4, pp. 289--293, Oct. 1958.

\bibitem{Gelfand2}
S.~I. Gel'fand and M.~S. Pinsker, ``Coding for channels with random
  parameters,'' \emph{Probl. Contr. Inf. Theory}, vol.~9, no.~1, pp. 19--31,
  1980.

\bibitem{Cizer}
I.~Csisz\'ar, ``The method of types,'' \emph{IEEE Trans. Inf. Theory}, vol.~44,
  no.~6, pp. 2505--2523, Oct. 1998.

\bibitem{Taleter}
E.~Telatar, ``Capacity of multi-antenna gaussian channels,'' \emph{European
  transactions on telecommunications}, vol.~10, no.~6, pp. 585--595, 1999.

\bibitem{Keshet}
G.~Keshet, Y.~Steinberg, and N.~Merhav, \emph{Channel Coding in the Presence of
  Side Information}.\hskip 1em plus 0.5em minus 0.4em\relax \emph{Foundations
  and Trends in Communications and Information Theory}, NOW Publishers,
  Hanover, MA, USA. vol. 4, Issue 6, pp. 1-144, 2007.

\bibitem{GoldsmithDan}
D.~Goldsmith, \emph{Fading Channels with Transmitter Side Information}.\hskip
  1em plus 0.5em minus 0.4em\relax MSc Thesis, EE Department, Technion-Israel
  Institute of Technology, Haifa, Israel, Dec. 2004.

\bibitem{Wyner}
A.~D. Wyner, ``A bound on the number of distinguishable functions which are
  time-limited and,'' \emph{SIAM J. Appl. Math.}, vol.~24, no.~3, pp. 289--297,
  May 1973.

\bibitem{Liang}
Y.~Liang, V.~H. Poor, and S.~Shamai, \emph{Information Theoretic
  Security}.\hskip 1em plus 0.5em minus 0.4em\relax Foundations and Trends in
  Communications and NOW Publishers, Hanover, MA, USA, 2009.

\bibitem{Wyner11}
A.~D. Wyner, ``The wire-tap channel,'' \emph{Bell Syst. Tech. J.}, vol.~54,
  no.~8, pp. 1355--1387, 1975.

\bibitem{ChenVinck}
Y.~Chen and H.~Vinck, ``Wiretap channel with side information,'' \emph{IEEE
  Trans. on Inf. Theory}, vol.~54, no.~1, pp. 395--402, Jan. 2008.

\bibitem{cover}
T.~M. Cover and J.~A. Thomas, \emph{Elements of Information Theory}.\hskip 1em
  plus 0.5em minus 0.4em\relax Wiley Series in Telecommunications and Signal
  Processing, 2nd Edition, 2006.

\bibitem{Neri2}
N.~Merhav, D.~Guo, and S.~Shamai, ``Statistical physics of signal estimation in
  {G}aussian noise: theory and examples of phase transitions,'' \emph{IEEE
  Trans. Inf. Theory}, vol.~56, no.~3, pp. 1400--1416, Mar. 2010.

\bibitem{Tul2}
A.~M. Tulino and S.~Verd\'u, ``Random matrix theory and wireless
  communications,'' \emph{Foundations and Trends In Communications and
  Information Theory}, vol.~1, no.~1, pp. 1--184, Jan. 2004.

\bibitem{SilversteinBai}
J.~W. Silverstein and Z.~D. Bai, ``On the empirical distribution of eigenvalues
  of a class of large dimensional random matrices,'' \emph{Journal of
  Multivariate Analysis}, vol.~54, no.~2, pp. 175--192, 1995.

\bibitem{Dembo}
A.~Dembo and O.~Zeitouni, \emph{Large Deviations Techniques and
  Applications}.\hskip 1em plus 0.5em minus 0.4em\relax Springer, 1998.

\bibitem{Hollander}
F.~Den~Hollander, \emph{Large Deviations}.\hskip 1em plus 0.5em minus
  0.4em\relax American Mathematical Society (Fields Institute Monographs),
  2000.

\bibitem{NeriMono}
N.~Merhav, ``Statistical physics and information theory,'' \emph{Foundations
  and Trends in Communications and Information Theory}, vol.~6, no. 1-2, pp.
  1--212, Dec. 2010.

\bibitem{Bruijn}
N.~G. De~Bruijn, \emph{Asymptotic Methods in Analysis}.\hskip 1em plus 0.5em
  minus 0.4em\relax Dover Publications, Inc. New York, 1981.

\end{thebibliography}
\end{document}